\documentclass[12pt,onecolumn,journal,draftclsnofoot,]{IEEEtran}
\usepackage{graphicx}
\usepackage{float}
\usepackage{epstopdf}
\usepackage[cmex10]{amsmath}
\usepackage{array}
\usepackage{cite}
\usepackage[utf8x]{inputenc}
\usepackage{amssymb}
\usepackage{stackrel}
\usepackage{booktabs,multirow}
\usepackage{etoolbox}
\makeatletter
\patchcmd{\@makecaption}
{\scshape}
{}
{}
{}
\makeatother
\usepackage{amsthm}
\usepackage{amsmath,mathtools}

\usepackage{adjustbox}

\newtheorem{Proposition}{Proposition}
\newtheorem{Theorem}{Theorem}
\newtheorem{Remark}{Remark}
\newtheorem{Definition}{Definition}
\newtheorem{Example}{Example}
\usepackage[T1]{fontenc}
\usepackage{cuted}
\usepackage{textcomp}

\begin{document}

\title{On Advisability of Designing Short Length QC-LDPC Codes Using Perfect Difference Families}
\author{Alireza Tasdighi and Mohammad-Reza Sadeghi
\thanks{A. Tasdighi and M. R. Sadeghi are with the Department
of Mathematics \& Computer Science, Amirkabir University of Technology, Tehran, Iran (e-mails: a.tasdighi@aut.ac.ir, msadeghi@aut.ac.ir)}}
\maketitle

\begin{abstract}
A simple and general definition of {\em quasi cyclic low-density parity-check} (QC-LDPC) codes which are constructed based on {\em circulant permutation matrices} (CPM) is proposed. As an special case of this definition, we first represent one type of so called combinatorially designed {\em multiple-edge protograph} codes. The code construction is mainly based on {\em perfect difference families} (PDF's) and is called Construction 1. Secondly, using the proposed Construction 1 along with a technique named as {\em column dispersion technique} (CDT), we design several types of multiple-edge CPM-QC-LDPC codes (codes with Construction 2) in a wide range of rates, lengths, girths and minimum distances. Parameters of some of these codes are reported in tables.
Also included in this paper are the multiplicities of short simple cycles of length up to $10$ in Tanner graph of our constructed codes. Our experimental results for short to moderate length codes show that although minimum distances of codes play an important role in waterfall region, the higher the number of short simple cycles is, the better (sharper) the waterfall is. The performances of many codes such as WiMAX, PEG, array, MacKay, algebraic and combinatorial, and also, symmetrical codes have compared with our constructed codes. Based on our numerical results and regardless of how a code is constructed, those with higher number of short simple cycles and higher minimum distances, have better waterfalls. As, high number of short simple cycles cause error floor phenomenon, we show that our constructed codes based on applying CDT on PDF's have a property that we can gradually increase their number of short simple cycles to outperform many of the codes in the literature whilst, preventing them to have very high number of short cycles and thus avoiding undesirable error floors. 
\end{abstract}



\maketitle

\section{Introduction}
\IEEEPARstart{A}{}fter rediscovery of low-density parity-check (LDPC) codes in the late 1990’s \cite{DMacKay1}, a great deal of research effort has been expended in design and construction of these
codes. Among designing all branches of LDPC codes, quasi cyclic (QC) codes have been allocated much more efforts of the researchers. Indeed, due to their capacity-achieving and also their intrinsic property of practically implementable decoding algorithm, QC-LDPC codes have been adopted as the standard codes for various next generations of communication systems \cite{WiMax1}. Applications of these codes to storage systems, such as flash memories, have now been seriously investigated \cite{CSchoeny1}.

Of famous methods in constructing QC-LDPC codes are, {\em graph-theoretic} and {\em protograph-based} methods \cite{Thrope1,Divsalar1,XHu1,karimi1,bocharova1,HPark1,HPark2,SRanganathan1,ATasdighi1,ATasdighi2}, as well as, {\em combinatoric-based} and {\em algebraic} methods \cite{JFan1,MFossorier1,sullivan1,hagiwara1,BVasic1,Lally1, JLi1,LZhang1,GZhang1}. Recently, designing short to moderate length QC-LDPC codes based on circulant permutation matrices (CPM's) have found much attraction among the researchers\cite{hagiwara1,SRanganathan1,ATasdighi1,ATasdighi2,karimi1,bocharova1,MBaldi2}. In almost all of these cited works, they attempt to find shortest possible length of CPM-QC-LDPC codes, given an specific {\em girth} and degree distribution. Short to moderate length codes may not approach very closely to the {\em Shannon limits}, however, one advantage of designing such codes is lowering the encoding and decoding complexity and so latency.

Another area in investigating QC-LDPC codes is finding their {\em minimum distances}\cite{HLiu1,Smarandache1,MBaldi1,BVasic1}. Minimum distance plays a direct and significant role in {\em bit-error-rate} (BER) performances of the codes. Simulation results exhibit high minimum distance of a code affects its performance even under sub-optimal {\em belief propagation} (BP) decoding algorithms. In fact, given a fixed length, degree distribution and girth of a QC-LDPC code, the higher minimum distance of the code the lower BER curve of the code is. For single-edge protograph codes it has established that the minimum
Hamming distance is always upper bounded by $(d_v +1)\,!$ \cite{MFossorier1,Smarandache2}. In addition, Smarandache {\em et al.} \cite{Smarandache1} showed that for a multiple-edge CPM-QC-LDPC code of maximum column weight $d_v\in \mathrm{N}$, the maximum reachable minimum distance could be even superior to $(d_v +1)\,!$. However, looking at numerical results reported for minimum distances of single-edge protograph codes in \cite{bocharova1,ATasdighi1,ATasdighi2}, as well as, considering our new experimental results for multiple-edge protograph codes, it seems that, as long as we try to find the shortest possible length of a CPM-QC-LDPC code (given a fixed girth and degree distribution), its minimum distance tends to be decreased. 

In this work we put our concentration on investigating other features of CPM-QC-LDPC codes of short to moderate lengths that may improve BER of the codes and thus compensating the lack of low minimum distance of such codes. Roughly speaking, we observed that if multiplicities of short simple cycles in Tanner graph of a code are high enough (approximately higher than the average possible ones, given a fixed length, degree distribution and minimum distance), then, BER curve of the code will be much lower than another code with similar parameters but with considerable smaller number of simple cycles. 

To construct short length CPM-QC-LDPC codes with high number of short simple cycles, we took benefits of {\em perfect difference families} (PDF's) and {\em quasi perfect difference families} (QPDF's). We regenerate one previously known of this type of the codes with girth 6, as well as, we introduce a new flexible design of such codes with girth at least 6.
Due to the existence of high number of {\em inevitable} short cycles in Tanner graph of our constructed codes, we are able to extend the length of our primary constructed code with short length to the moderate length one by increasing $N$, where, $N$ is the {\em lifting degree} of our primary code. The inevitable short simple cycles will not vanish by increasing $N$, instead, their multiplicities will linearly increase by enlarging $N$. Although minimum distances of our codes are small, due to the high number of their simple cycles, BER performances of these codes outperform their counterparts in the literature. Table \ref{Tab3} has summarized the details of many of compared codes which are decoded in an Additive White Gaussian Noise (AWGN) channel under {\em Sum-Product} (SP) algorithm.

High number of short cycles in Tanner graph of LDPC codes may cause {\em error floor} phenomenon in BER curve of the codes. Thanks to our flexible design of CPM-QC-LDPC codes (called Construction 2) that use PDF's as their shift-values of CPM's, we show that starting from a primary QC code and without decreasing its girth and minimum distance, it is possible to design a range of QC codes, where, their number of simple cycles could be gradually increased. Our software simulations reveal that with a subtly design of the base matrix of the codes, and simultaneously, by gradually augmentation of their short cycles, it is possible to reach a target BER (down to $10^-9$ for short to moderate length codes) without observing error floor and these codes still outperform almost all of their counterpart codes in waterfall region. Figures of some of the simulated codes are added in section \ref{ENER}.  

The rest of the paper is organized as follows. Section \ref{Sec2}
presents a primary definition of difference families, theorems to show their existence, definition to CPM based QC-LDPC codes and a general formula for calculating expected values of simple cycles in Tanner graph of a regular code. In section \ref{Sec3} we first represent one previously known construction of CPM-QC-LDPC codes using PDF's by name Construction 1. In the following and using Construction 1, we introduce a flexible
method for constructing a CPM-QC-LDPC code name Construction 2. Section \ref{ENER} contains all of the experimental results such as examples, tables, figures, and our justification in advantages of using PDF's in designing short length QC-LDPC codes. Section \ref{Sec5} concludes the paper.
\section{Preliminaries}\label{Sec2}
\subsection{Perfect Difference Families}
\begin{Definition}\label{Def1}
\cite{CColbourn1} Let $G$ be a group of order $v$. A collection $\lbrace B_1, \cdots,B_t \rbrace$ of $k$-subsets of $G$ form a
$(v, k,\lambda)$ {\em difference family} (DF) if every non-identity element of $G$
occurs $\lambda$ times in $\Delta B_1 \cup \Delta B_2 \cup \cdots \cup \Delta B_t$, where, $\Delta B = \lbrace b_i - b_j |b_i,b_j \in B;\; i, j = 1, . . . ,k;\; i \neq j \rbrace$. The sets $B_i$ are base blocks.
\end{Definition}

If group $G$ in previous definition is considered to be additive group $Z_v =
\lbrace0, 1, . . . , v - 1\rbrace$, then $t$ $k$-element subsets of $Z_v$, $B_i =
\lbrace b_{i1}, b_{i2}, . . . , b_{ik}\rbrace,\; i = 1, 2, . . . , t,\; b_{i1} < b_{i2} <\cdots < b_{ik}$, form
a $(v, k, \lambda)$ {\em cyclic difference family} (CDF) if every non-zero
element of $Z_v$ occurs $\lambda$ times among the differences $b_{im}-b_{in}$,
$i = 1, 2, . . . , t,\; m \neq n,\; m, n = 1, 2, . . . , k$. We, hereafter and without loss of generality will work with CDF and refer to it simply by DF. In addition, if $v = k(k - 1)t + 1$, then $t$ blocks $B_i = \lbrace b_{i1}, b_{i2}, . . . , b_{ik}\rbrace$ form a {\em perfect} $(v, k, 1)$ difference family (PDF) over $Z_v$ if the $tk(k - 1)/2$ differences $b_{im} - b_{in} (i = 1,\cdots,t,\; 1 \leq m < n \leq k)$ cover the set $\lbrace 1, 2,\cdots, (v - 1)/2\rbrace$. If instead, they cover the set $\lbrace 1,2,\cdots,(v - 3)/2 \rbrace \cup \lbrace(v + 1)/2\rbrace$, the difference family is {\em quasi-perfect}, where we denote it by QPDF.

Considering references \cite{JDinitz1,JHuang1,RMathon1,HPark2,CColbourn1} the following two theorems are well-known for the existence of DF and/or PDF:
\begin{Theorem}\label{Theo1}
The existence of $(k(k -1)t+1, k, 1)$ DFs is given as:
\begin{itemize}
\item[1)] There exists a $(6t + 1, 3, 1)$ DF for all $t \geq 1$.
\item[2)] A $(12t + 1, 4, 1)$ DF exists for all $1 \leq t \leq 1000$, except for $t=2$.
\item[3)] A $(20t + 1, 5, 1)$ DF exists for $1 \leq t \leq 50$ and $t \neq
16, 25, 31, 34, 40, 45$.
\end{itemize}
\end{Theorem}
\begin{Theorem}\label{Theo2}
The existence of $(k(k - 1)t+1, k, 1)$ PDFs is given as:
\begin{itemize}
\item[1)] A $(6t + 1, 3, 1)$ PDF exists, if and only if, $t \equiv 0\; \text{or}\; 1 \mod 4$. 
\item[2)] A $(12t + 1, 4, 1)$ PDF exists for $t = 1, 4 \leq t \leq 1000$.
\item[3)] $(20t + 1, 5, 1)$ PDFs are known for $t = 6, 8, 10$ but for no other values in $1 \leq t \leq 50$.
\item[4)] There is no $(k(k - 1)t + 1, k, 1)$ PDF for $k \geq 6$.
\end{itemize}
\end{Theorem}
Just like \cite{HPark2} and since there are no PDFs for $k \geq 6$, as well as, no sufficiently
many PDFs for $k = 5$ from Theorem \ref{Theo2}, we put our concentration on investigation of two cases
 $(v,k=3,1)$ and $(v,k=4,1)$. Moreover, given $k=3,4$, for those cases of $t$'s that a PDF does not exist, we focus on an existing QPDF, instead. Some of the PDF's/QPDF's are provided in section \ref{ENER} (Tables \ref{Tab1} and \ref{Tab2}) using {\em Skolem sequences} and {\em gracefully labelling prisms} technique that are presented in \cite{JDinitz1,JHuang1,RMathon1}. Interested reader is also referred to \cite{CColbourn1}.
\subsection{CPM-QC-LDPC Codes And Inevitable Cycles}
Let $\mathcal{C}$ be a binary LDPC code whose parity-check matrix $H$
is an $m\times n$ array of matrices as follows:
\begingroup\fontsize{8.5pt}{11pt}\begin{align}\label{Rela1}
H=\left[\begin{array}{cccc}
H_{00}&H_{01}&\cdots &H_{0(n-1)}\\
H_{10}&H_{11}&\cdots &H_{1(n-1)}\\
\vdots &\vdots &\ddots &\vdots \\
H_{(m-1)0}&H_{(m-1)1}&\cdots &H_{(m-1)(n-1)}\\
\end{array}\right],
\end{align}\endgroup
where, $H_{ij}\;(0\leq i \leq m-1 \;\&\; 0\leq j \leq n-1)$ is either a circulant matrix that each row is a cyclically right shift of the row above it, or, a zero square matrix of size $N$. As there are $m$ row blocks in $H$, such an LDPC code is named as $m$-{\em level QC-LDPC code}. The bipartite graph which is constructed based on bi-adjacency matrix $H$ is the so-called {\em Tanner graph} of $\mathcal{C}$. Within our context, we assume that each one of the non-zero circulant matrices $H_{ij}$ is a combination (a matrix summation) of some circulant permutation matrices (CPM's) as follow:
$$ H_{ij}=I^{p^{1}_{ij}}+I^{p^{2}_{ij}}+\cdots +I^{p^{l}_{ij}},$$
where, $p^{r}_{ij}\in \lbrace 0,1,\cdots,N-1\rbrace$ and $\;p^{r}_{ij}\neq p^{r^\prime}_{ij}$ for $\;1\leq r<r^\prime\leq l,\;l\in \mathrm{N}$. $I^{p^{r}_{ij}}$ is a CPM of size $N$ which its first row is $p^{r}_{ij}$ cyclically right shifts of an identity matrix of size $N$. Let 
$ L=\max\limits_{i,j} \lbrace l| l\text{ is the number of CPM's that constitute } H_{ij}\rbrace$
so, such an LDPC code is known as {\em $m$-level Type-$L$ CPM-QC-LDPC code} and the number $N$ is {\em lifting degree} of the code. Note that our definition of CPM-QC-LDPC codes covers all types of {\em single-edge protograph codes} (\hspace*{-0.1cm}\cite{ATasdighi1}) and {\em multiple-edge protograph codes} (\hspace*{-0.1cm}\cite{HPark1}) that are produced by lifting each edge of their corresponding base graphs. 

The matrix $H$ corresponds to the following exponent matrix
\begingroup\fontsize{8.5pt}{11pt}\begin{align}\label{Rela2}
P=\left[\begin{array}{cccc}
\vec{P}_{00}&\vec{P}_{01}&\cdots &\vec{P}_{0(n-1)}\\
\vec{P}_{10}&\vec{P}_{11}&\cdots &\vec{P}_{1(n-1)}\\
\vdots &\vdots &\ddots &\vdots \\
\vec{P}_{(m-1)0}&\vec{P}_{(m-1)1}&\cdots &\vec{P}_{(m-1)(n-1)}\\
\end{array}\right],
\end{align}\endgroup
where, $\vec{P}_{ij}=\left( p^{1}_{ij},p^{2}_{ij},\cdots ,p^{l}_{ij}\right)$, corresponding to any non-zero circulant matrices $H_{ij}$ in equation (\ref{Rela1}), as well as, $\vec{P}_{ij}=(-\infty)$ if $H_{ij}$ is a zero matrix. It is obvious that the parity-check matrix $H$ will uniquely be identified if its exponent matrix $P$ is known. 
If $L=1$, then each $\vec{P}_{ij}$ is an integer or an infinity symbol where the resulting $m$-level CPM-QC-LDPC code is a single-edge protograph code. If none of the elements in $P$ are infinity (i.e., $\mathcal{C}$ is a fully-connected single-edge protograph code), then, it is well-known \cite{MFossorier1} that the necessary and sufficient condition for the existence of a cycle of length $2k^{\prime}$ in the Tanner graph
of $\mathcal{C}$, corresponding to $H$, is \vspace*{-0.16cm}
\begingroup\fontsize{8.5pt}{11pt} \begin{equation}\label{Rela3}
\sum_{i=0}^{k^{\prime}-1} \left( P_{m_{i}n_{i}} - P_{m_{i}n_{i+1}} \right) = 0  \mod N , 
\end{equation} \endgroup  
where $n_{k^{\prime}}=n_{0}$, $m_{i} \neq m_{i+1}$, $n_{i} \neq n_{i+1}$ and $P_{m_{i}n_{i}}$ is the $(m_i n_i)$-th entry of $P$. In the case that $L=1$ and $P$ contains infinity elements, equation (\ref{Rela3}) is also reliable in determining a cycle of length $2k^{\prime}$ where, we just need to ignore those summations that contain even one infinity symbol. Indeed, as each infinity symbol is corresponded to a zero matrix $H_{ij}$ in bi-adjacency matrix $H$, the Tanner graph of $H$ has no edge between those vertices that are corresponded to the columns and rows of $H_{ij}$. Therefore, these type of summations will not result in a cycle of length $2k^{\prime}$. In the case $L>1$, there are some vectors $\vec{P}_{ij}$'s of lengths greater than one. It means that the resulting $m$-level CPM-QC-LDPC code is a multiple-edge protograph code. It has been shown (Lemma 1 in \cite{HPark1}) that with a slight difference, equation (\ref{Rela3}) is also true for the case of multiple-edge protograph codes. Lemma 1 in \cite{HPark1} provides a graphical representation of the new equation (\ref{Rela3}) for multiple-edge protograph codes, however, here we need to consider an equivalent version of this lemma, which is: \vspace*{-0.16cm}
\begingroup\fontsize{8.5pt}{11pt}  \begin{equation}\label{Rela4}
\sum_{i=0}^{k^{\prime}-1} \left( P^{r^{}_i}_{m_{i}n_{i}} - P^{r^{\prime}_i}_{m_{i}n_{i+1}} \right) = 0  \mod N , 
\end{equation}
\endgroup
where $n_{k^{\prime}}=n_{0}$, $r^{}_{i}\neq r^{\prime}_{i}$ if $n_{i} = n_{i+1}$, $ r^{\prime}_{i}\neq r^{}_{i+1}$ if $m_{i} =m_{i+1}$, $P^{r^{}_i}_{m_{i}n_{i}}$ is the $r^{}_i$-th entry of $\vec{P}_{m_{i}n_{i}}$ and $\vec{P}_{m_{i}n_{i}}$ is the $(m_i n_i)$-th entry of $P$. 
Although, equation (\ref{Rela4}) defines a necessary and sufficient condition for the existence of a cycle of length $2k^{\prime}$ in the Tanner graph of $\mathcal{C}$ when $\mathcal{C}$ is an arbitrary $m$-level Type-$L$ CPM-QC-LDPC code, we again need to ignore those summations that contain even one infinity symbol.

Similar to the definition 1 in \cite{HPark1}, here we define an inevitable cycle which is needed in our context.
\begin{Definition}\label{Def2}
 An inevitable cycle of length $2k^{\prime}$ in Tanner graph of an $m$-level Type-$L$ CPM-QC-LDPC code of lifting degree $N$ is a cycle that always appear in Tanner graph of the code regardless of how high the lifting degree $N$ is. That is the equation (\ref{Rela4}) is always equal to zero, regardless of the size of $N$. 
\end{Definition}
\begin{Example}\label{Ex1}
Let $m=3$, $n=4$ and $N=13$. Consider the matrix:
\begingroup\fontsize{9.5pt}{11pt} \begin{align*}
P=\left[\begin{array}{cccc}
(0,1,8) & (-\infty) & (0) & (-\infty)\\
(-\infty) & (8,12) & (0,4) & (-\infty)\\
(-\infty) & (5) & (-\infty) & (4,9,10)
\end{array}\right],
\end{align*}
\endgroup
where, $P$ is an exponent matrix for a $(3,4)$-regular $3$-level Type-$3$ CPM-QC-LDPC code $\mathcal{C}$. In this case, below equations I and II of type in equation (\ref{Rela4}), respectively indicate the existence of inevitable cycles of lengths $4$ and $6$ in Tanner graph of $\mathcal{C}$.
\begingroup\fontsize{10pt}{9pt}
\begin{itemize}
\item[I.] $P^1_{11}-P^1_{12}+P^2_{12}-P^2_{11}=(8-0)+(4-12)=0$ 
\item[II.] $P^1_{00}-P^2_{00}+P^3_{00}-P^1_{00}+P^2_{00}-P^3_{00}=(0-1)+(8-0)+(1-8)=0$
\end{itemize}\endgroup
\end{Example}
Before terminating this part, we briefly notify some parameters of a linear code $\mathcal{C}$. Regardless whether a cycle is inevitable or not, the {\em girth} ($\mathbf{g}$) of a code $\mathcal{C}$ is given to be the length of shortest cycle in Tanner graph of the code. Also, {\em minimum Hamming distance} ($d_\mathbf{min}$) of a linear code $\mathcal{C}$ is the minimum weight of all codewords in $\mathcal{C}$, where, weight of each  codeword is the number of non-zero elements in a vector representation of that codeword. {\em Dimension} ($\mathbf{Dim}$) of a linear code $\mathcal{C}$ is the number of linearly independent codewords in vector space of the code. Finally, in the case where $\mathcal{C}$ is a CPM-QC-LDPC code of lifting degree $N$ with an $m\times n$ exponent matrix $P$, the {\em Length} ($\mathbf{Len}$) of $\mathcal{C}$ is equal to $nN$.
\subsection{Expected Number of Simple Cycles in Regular Tanner Graphs} 
Let $H$ be a regular parity-check matrix of a code $\mathcal{C}$ with column-weight $d_v$ and row-weight $d_c$. Suppose that the  number of columns and the number of rows of $H$ are $n^{\prime}$ and $m^{\prime}$, respectively. It is shown \cite{XJiao1} that the number of simple cycles of length $2k^{\prime}$ which are included in Tanner graph of $\mathcal{C}$ have a Poisson's distribution with expected number:
\begingroup\fontsize{8pt}{11pt} \begin{equation}\label{Rela5}
\begin{gathered} exp(C2k^{\prime})=\\ \hfill
\binom{m^{\prime}}{k^{\prime}} \binom{n^{\prime}}{k^{\prime}} \frac{{k^{\prime}\,!}\:(k^{\prime}-1)\,!}{2} \frac{\left[ d_v (d_v -1)\right]^{k^{\prime}}\left[ d_c (d_c -1)\right]^{k^{\prime}}}{\vert E\vert \times (\vert E\vert -1) \times \cdots \times (\vert E\vert -2k^{\prime}+1)},\hfill \end{gathered}
\end{equation}\endgroup
where, $\vert E\vert$ is the number of edges of Tanner graph of $\mathcal{C}$. Note that $\vert E\vert=d_v n^{\prime}=d_c m^{\prime}$ in a $(d_v,d_c)$-regular Tanner graph. Table \ref{Tab4} in section \ref{ENER} contains the expected values of simple cycles that they are in Tanner graphs of the investigated codes with different lengths and degree distributions.
\section{Code Construction}\label{Sec3}
In this section, we represent one type of previously known QC-LDPC code constructed based on DF's \cite{BVasic1} and also PDF's \cite{HPark2} as a special case of our $m$-level Type-$L$ CPM-QC-LDPC codes.  However, first we introduce a new simple technique named {\em column dispersion technique} (CDT). Using CDT, it is possible to design new ensembles of $m$-level Type-$L^*$ CPM-QC-LDPC codes with girth $\mathbf{g}^*$ given a $1$-level Type-$L$ CPM-QC-LDPC codes with girth $\mathbf{g}$, where, $m>1$, $L^*\leq L$ and $\mathbf{g}^*\geq \mathbf{g}$.
\begin{Definition}\label{Def3}
Let $\vec{P}_{ij}=\left( p^1_{ij}, p^2_{ij},\cdots, p^l_{ij}\right)$ be an arbitrary element of exponent $P$ in equation (\ref{Rela2}) and $m>1 \; (m\in \mathrm{N})$. $m$-column dispersion ($m$-CD) of $\vec{P}_{ij}$ partitions all the elements of this vector into $m$ pieces of smaller vectors $\vec{P}^1_{ij}$ to $\vec{P}^m_{ij}$ first, and secondly, put them in a column of size $m$, respectively from $\vec{P}^1_{ij}$ down to $\vec{P}^m_{ij}$. Also, partitioning $\vec{P}_{ij}$ doesn't necessarily mean a conventional partitioning. Instead, there may be one or more empty sets in this partition, that is, vector(s) with one infinity symbol, only. 
\end{Definition}
\begin{Example}\label{Ex2}
Let $\vec{P}_{ij}=(0, 3,4,11,17)$ be an arbitrary element of exponent matrix $P$. Below column-blocks $D_1$ and $D_2$ are two different $5$-CD for $\vec{P}_{ij}$
\begingroup\fontsize{8pt}{11pt}$$D_1=\left[\begin{array}{c}
(0)\\
(3,11)\\
(-\infty)\\
(17)\\
(4)
\end{array}\right] \;\;\mbox{and} \;\;\;\; D_2=\left[\begin{array}{c}
(0)\\
(3)\\
(4)\\
(11)\\
(17)
\end{array}\right].$$\endgroup
\end{Example}
\begin{Proposition}\label{Propo1}
Let $L, L^*, m,n, N \in \mathrm{N}$, $m>1$ and $1\leq L^*\leq L$. Suppose that $P$ is a $1\times n$ exponent matrix of a $1$-level Type-$L$ CPM-QC-LDPC code $\mathcal{C}$ with lifting degree $N$, girth $\mathbf{g}$ and minimum distance $d_{\mathbf{min}}$. If for each $j\in \lbrace 0,1,\cdots, n-1\rbrace$, $j$-th column of $m\times n$ exponent matrix $P^*$ is produced by applying $m$-CDT on $\vec{P}_{1j}$, then $P^*$ results in an $m$-level Type-$L^*$ CPM-QC-LDPC code $\mathcal{C}^*$ with lifting degree $N$ girth ${\mathbf{g}}^*$ and minimum distance $d^*_{\mathbf{min}}$, where, ${\mathbf{g}}^* \ge \mathbf{g}$ and $d^*_{\mathbf{min}} \ge d_{\mathbf{min}}$.
\end{Proposition}
\begin{proof}
From the definition of column dispersion of entry $\vec{P}_{1j}$ ($0\leq j\leq n-1$) it is obvious that by partitioning each one of these vectors into smaller vectors, their maximum length $L^*$ is less than or at most equal to the maximum length of the vectors in $P$. To see that $\mathbf{g}^*\geq \mathbf{g}$, it is sufficient to show that equation (\ref{Rela4}) is not equal to zero (module $N$) for those cycles of Tanner graph of code $\mathcal{C}^*$ which have lengths less than $\mathbf{g}$. Let $2k^{\prime}<\mathbf{g}$. Also, let ${P^*}^{r^{}_i}_{m_{i}n_{i}}$ and ${P^*}^{r^{\prime}_i}_{m_{i}n_{i+1}}$ be respectively, the $r^{}_i$ and $r^{\prime}_i$ elements ($r^{}_i,r^{\prime}_i\leq L^*$) of vectors $\vec{P^*}_{m_{i}n_{i}}$ and $\vec{P^*}_{m_{i}n_{i+1}}$, where, $\vec{P^*}_{m_{i}n_{i}}$ and $\vec{P^*}_{m_{i}n_{i+1}}$ are  the $(m_{i}n_{i})$ and $(m_{i}n_{i+1})$ entries of $P^*$. Then;
\begingroup\fontsize{8.5pt}{11pt}\begin{align}\label{Rela6}
\sum_{i=0}^{k^{\prime}-1} \left( {P^*}^{r^{}_i}_{m_{i}n_{i}} - {P^*}^{r^{\prime}_i}_{m_{i}n_{i+1}} \right) =\sum_{i=0}^{k^{\prime}-1} \left( {P}^{\bar{r}^{}_i}_{1n_{i}} - {P}^{\bar{r}^{\prime}_i}_{1n_{i+1}} \right) 
\end{align} \endgroup
where, ${P}^{\bar{r}^{}_i}_{1n_{i}}$ and ${P}^{\bar{r}^{\prime}_i}_{1n_{i+1}}$ are respectively, the $\bar{r}^{}_i$ and $\bar{r}^{\prime}_i$ elements ($\bar{r}^{}_i,\bar{r}^{\prime}_i\leq L$) of vectors $\vec{P}_{1n_{i}}$ and $\vec{P}_{1n_{i+1}}$ in $P$. Moreover, if $r^{}_i\neq r^{\prime}_i$ ($r^{\prime}_i\neq r^{}_{i+1}$) when $n_i=n_{i+1}$ ($m_i=m_{i+1}$), then, $\bar{r}^{}_i\neq \bar{r}^{\prime}_i$ ($\bar{r}^{\prime}_i\neq\bar{r}^{}_{i+1}$). As the right hand side of equation (\ref{Rela6}) is always unequal to zero module $N$ ($\mathcal{C}$ is of girth $\mathbf{g}>2k^{\prime}$), so, the left hand side of this equation is not equal to zero module $N$. To demonstrate that $d^*_\mathbf{min} \geq d_\mathbf{min}$, we first consider the parity check matrices $H$ and $H^*$ of type in equation (\ref{Rela1}) that are respectively lifted based on exponent matrices $P$ and $P^*$ given a lifting degree $N$. It is sufficient to show that if the set $\lbrace c^*_1,c^*_2,\cdots,c^*_{d^*}\rbrace$ ($d^*_\mathbf{min}\leq d^*$) forms a set of dependent columns of $H^*$ then the set $\lbrace c_1,c_2,\ldots,c_{d^*}\rbrace$ constitutes a set of dependent columns of $H$ where both columns $c^*_i$ and $c_i$ ($1\leq i \leq d^*$) are in a similar location, respectively, in $H^*$ and $H$\footnote{By a set of dependent columns we mean that the summation of all columns in that set is equal to an all zero column vector module $2$.}. Note that $H$ ($H^*$) consists of $1$ ($m$) row block(s) and $n$ ($n$) column blocks. Since $P^*$ is constructed by applying $m$-CDT on $P$, shift values in each column of $P^*$ (correspondingly the CPM's in each column block of $H^*$) are exactly the same as those ones in each column (column block) of $P$ ($H$). We consider two facts: 1) if two columns $c^*_{i_1}$ and $c^*_{i_2}$ ($i_1\neq i_2$) intersect in a common row in $H^*$, their corresponding columns $c_{i_1}$ and $c_{i_2}$ will intersect in a common row as a part of row block $1$ in $H$. Note that since the girths of $\mathcal{C}^*$ and $\mathcal{C}$ are greater than or equal to $6$, so, each two different columns of their parity check matrices will intersect in at most one common row. 2) assuming that two columns $c^*_{i_1}$ and $c^*_{i_2}$ ($i_1\neq i_2$) intersect in a common row $r_j$ ($j\in\lbrace 1,2,\cdots,N\rbrace$) of a row block number $r$ ($1\leq r\leq m$) of $H^*$. So, none of these columns intersect in another column $c^*_{i_3}$ in prespecified common row $r_j$ of any row block $r^{\prime}$ ($1\leq r^{\prime}\leq m\;\&\;r\neq r^{\prime}$). It is because, none of the CPM's in each column block of $H^*$ are the same. Facts (1) and (2) guarantee that if the summation \begingroup\fontsize{8.5pt}{11pt}$\sum_{l=1}^{d^*}c^*_{l}\;$\endgroup is equal to a zero vector module $2$ then the summation \begingroup\fontsize{8.5pt}{11pt}$\sum_{l=1}^{d^*}c_{l}\;$\endgroup is also equal to a zero vector module $2$. This means that for any set of dependent columns of $H^*$ of cardinality $d^*$ there is a set of dependent columns of $H$ of the same cardinality. However, there may be set's of dependent columns of $H$ with smaller cardinality. Thus $d^*_\mathbf{min} \geq d_\mathbf{min}$.
\end{proof}
\subsection{Construction 1}
Given a $(v,k,1)$ DF with $t$ blocks $B_{i}\;(1\leq i\leq t)$, each of size $k$ and let $P$ be a $1\times t$ exponent matrix, where, entries of vector $\vec{P}_{0,j}\;(0\leq j\leq t-1)$ are exactly those integer numbers in $B_{j+1}$. If $N\in \mathrm{N}$ is taken to be greater than all of the elements in $B_i$'s, then $P$ results in a $(d_v=k,d_c=tk)$-regular $1$-level Type-$k$ CPM-QC-LDPC code $\mathcal{C}$ of lifting degree $N$ and length $tN$ where its Tanner graph has no inevitable cycle of length $4$ \cite{BVasic1}. In order to prevent those remaining 4-cycles that are not inevitable (i.e., in order to $\mathcal{C}$ has a girth $\mathbf{g}\geq 6$), we have to pick $N$ in a way that equation (\ref{Rela4}) is not equal to zero module $N$ when $k^{\prime}=2$. It is known \cite{HPark2} that if our considered DF is a PDF, then $P$ with any lifting degree $N\geq v$ result in a code of girth at least 6. In addition, from the definition of QPDF it is clear that if we consider a QPDF instead of PDF, then  $P$ with any lifting degree $N=v-1,v+1,v+2,\ldots$ results in a code of girth at least 6. Note that Construction 1 was first introduced in \cite{BVasic1} considering BIBD's and later in \cite{HPark2}, so we leave its prove for the sake of brevity. Also, Construction 1 results in QC-LDPC codes with $k+1\leq d_\mathbf{min}\leq 2k$ \cite{BVasic1}. Tables \ref{Tab1} and \ref{Tab2} in section \ref{ENER} contains many PDF's (and/or QPDF's) along with a range of accurate lifting degrees that are used to construct girth $6$ $1$-level Type-$k$ CPM-QC-LDPC codes, when $k=3,4$. Also included in these tables are the minimum distances of such codes for which the lifting degrees are given to be the minimum possible ones.
\begin{Example}\label{Ex3}
Let $(v=37,k=3,\lambda=1)$ be QPDF in Table I where $t=6$. Based on Construction 1, below $1\times 6$ exponent matrix:
\begingroup\fontsize{9pt}{11pt}$ P=\left[ \begin{array}{@{}c@{}c@{}c@{}c@{}c@{}c@{}}
(0,1,8)\;\;\; & (0,6,19)\;\;\; & (0,9,12)\;\;\; & (0,10,14)\;\;\; & (0,11,16)\;\;\; & (0,15,17)
\end{array}\right]$ \endgroup
results in a $(d_v=3,d_c=18)$-regular $1$-level Type-$3$ CPM-QC-LDPC code with girth $6$, $d_\mathbf{min}=4$ and lifting degree $N=37, 39,40,\ldots\;$.
\end{Example}
\subsection{Construction 2}
Let $\mathcal{C}_1$ be the $1$-level Type-$L_1$ CPM-QC-LDPC code of lifting degree $N$, girth  $\mathbf{g}_1$ and minimum distance ${d_1}_\mathbf{min}$ constructed based on Construction 1. Subject to the Proposition \ref{Propo1} and considering $\mathcal{C}_1$, it is possible to construct a new $m$-level Type-$L_2$ CPM-QC-LDPC code $\mathcal{C}_2$ of lifting degree $N$, girth $\mathbf{g}_2$ and minimum distance ${d_2}_\mathbf{min}$ where, $m>1$, $1\leq L_2\leq L_1$, $\mathbf{g}_2\geq \mathbf{g}_1$ and ${d_2}_\mathbf{min} \geq {d_1}_\mathbf{min}$.
\begin{Example}\label{Ex4}
Let $P_1$ be the exponent matrix in Example \ref{Ex3}. Given $m=3$, $L_2=1$ and $N=37,39,40,\ldots$ , then, the following matrix
\begingroup\fontsize{9pt}{11pt} $$  P_2=\left[ \begin{array}{cccccc}
0 & 0 & 0 & 0 &0&0 \\
1 & 6 & 9 &10&11&15 \\
8 &19&12&14&16&17
\end{array}\right] $$ \endgroup
which is constructed by applying $3$-CDT (described in Proposition \ref{Propo1}) on $P_1$, is an exponent matrix to a $(3,6)$-regular $3$-level Type-$1$ CPM-QC-LDPC code with girth at least $6$ and minimum distance at least $4$. 
\end{Example}
\begin{Remark}\label{Rem1}
Proposition \ref{Propo1} shows that it is possible to preserve or probably improve the girth of the code that is constructed based on CDT without changing its lifting degree $N$. However, based on how we disperse the underlying primary exponent matrix to the desired (second) exponent matrix, there may exist integer numbers smaller than $N$ which taking them as accurate lifting degrees for second exponent matrix, result in CPM-QC-LDPC codes with girth at least $6$. For example, using a computer search it is easy to check that exponent matrix $P_2$ in Example \ref{Ex4} with lifting degrees $N\geq 20$ also results in girth $6$ codes.
\end{Remark}
\begin{Remark}\label{Rem2}
As it was previously expressed, codes constructed with Construction 1 have an upper bound minimum distance equal to $2k$ where $k$ was the column weight of the codes. However, codes constructed by Construction 2 (i.e., by applying Proposition \ref{Propo1} on the codes of Construction 1) may have larger minimum distances. For instance, we found out exponent matrix $P_2$ in Example \ref{Ex4} with lifting degree $20$ and $21$ results in a code with minimum distance\footnote{All the reported (upper bound or exact) minimum distances of simulated codes have been found by MAGMA \cite{MAGMA}.} equal to $8$ and $10$, respectively. 
\end{Remark}
\section{Numerical and Experimental Results}\label{ENER}
This section comprises numerical results to PDF's and QPDF's when $k=3,4$ (Tables \ref{Tab1} and \ref{Tab2}), our examples in constructing CPM-QC-LDPC codes of various lengths and rates using Constructions 1 and 2, tables to summarize the parameters of all of the simulated codes (Tables \ref{Tab3} and \ref{Tab4}) and several figures of BER performances of some of the simulated codes. Moreover, we put our justifications on advisability of designing CPM-QC-LDPC codes using PDF's as explanatory notes. 
\begin{table}[ht]
\setlength{\tabcolsep}{.7 pt}
\centering
\caption{$(v,k=3,\lambda=1)$ Difference Families That Result In $1$-level Type-III QC-LDPC Codes With Girth 6.
($N$ Is The Accurate Lifting Degree. $t$ Is The Number Of Blocks. $\mathbf{Dim}$ And $d_\mathbf{min}$ Are Respectively Dimension And Minimum Distance Of The Resulting Code Constructed By Smallest Accurate Lifting Degree $N$.)}
\vspace{-.7em}
\begingroup\fontsize{5 pt}{4 pt}
\begin{tabular}{|@{}c@{}|@{}c@{}|@{}l@{}|@{}c@{}|@{}c@{}|@{}c@{}|@{}l@{}|}
\hline
 $t$ & $v$ & ~~~~~~$N$ & $\mathbf{Dim}$ &$d_\mathbf{min}$ & Type & ~~~~~~~~~~~~Difference Set's \\
\hline
\hline  
$1$ & $7$ & $7,8,9,\ldots$ &  $3$ & $4$  &\textbf{PDF}& $\lbrace 0,1,3\rbrace$ \\
\hline
  $2$ & $13$ &$13,15,16,\ldots$&  $13$ & $4$  & \textbf{QPDF} & $\lbrace 0,1,4\rbrace,\lbrace 0,2,7\rbrace$ \\
\hline
  $3$ & $19$ & $19,21,22,\ldots$ &  $38$ & $6$  &\textbf{QPDF}& $\lbrace 0,1,5\rbrace,\lbrace 0,3,10\rbrace, \lbrace 0, 6,8\rbrace$ \\
\hline
 $4$ & $25$ & $25,26,27,\ldots$ &  $75$ & $4$  &\textbf{PDF}& $\lbrace 0,1,6\rbrace,\lbrace 0,4,12\rbrace,\lbrace 0,7,10\rbrace, \lbrace 0,9,11\rbrace$\\
\hline
  $5$ & $31$ &$31,32,33,\ldots$&  $124$ & $6$  &\textbf{PDF}&  $\begin{gathered} \lbrace 0,1,7\rbrace,\lbrace 0,5,15\rbrace,\lbrace 0,8,11\rbrace,\lbrace 0,9,13\rbrace , \hfill \\ \lbrace 0,12,14\rbrace\hfill \end{gathered} $\\
\hline
  $6$ & $37$ & $37,39,40,\ldots$ &  $185$ & $4$  & \textbf{QPDF} & $\begin{gathered} \lbrace 0,1,8\rbrace,\lbrace 0,6,19\rbrace,\lbrace 0,9,12\rbrace,\lbrace 0,10,14\rbrace ,\hfill \\ \lbrace 0,11,16\rbrace ,\lbrace 0, 15,17\rbrace \hfill \end{gathered}$ \\
\hline
  $7$ & $43$ & $43,45,46,\ldots$ &  $258$ & $6$  & \textbf{QPDF} & $\begin{gathered} \lbrace 0,1,9\rbrace,\lbrace 0,7,22\rbrace,\lbrace 0,10,12\rbrace,\lbrace 0,11,16\rbrace ,\hfill \\ \lbrace 0,13,19\rbrace ,\lbrace 0,14,18\rbrace,\lbrace 0,17,20\rbrace
  \hfill \end{gathered}$ \\
\hline
  $8$ & $49$ & $49,50,51,\ldots$ &  $343$ & $4$  &  \textbf{PDF} & $\begin{gathered} \lbrace 0,1,10\rbrace,\lbrace 0,8,24\rbrace,\lbrace 0,11,13\rbrace,\lbrace 0,12,17\rbrace , \hfill \\ \lbrace 0,14,20\rbrace ,\lbrace 0,15,22\rbrace,\lbrace 0,18,21 \rbrace, \lbrace 0,19,23 \rbrace \hfill \end{gathered}$ \\
\hline
  $9$ & $55$ & $55,56,57,\ldots$ &  $440$ & $4$  & \textbf{PDF} &$\begin{gathered} \lbrace 0,1,27\rbrace,\lbrace 0,2,14\rbrace,\lbrace 0,4,24\rbrace,\lbrace 0,5,15\rbrace , \hfill \\ \lbrace 0,6,23\rbrace ,\lbrace 0,7,18\rbrace,\lbrace 0,8,21\rbrace, \lbrace 0,9,25\rbrace, \hfill \\ \lbrace 0,19,22\rbrace
  \hfill \end{gathered}$ \\
\hline
  $10$ & $61$ & $61,63,64,\ldots$ &  $549$ & $6$  & \textbf{QPDF} & $\begin{gathered} \lbrace 0,2,3\rbrace,\lbrace 0,10,17\rbrace,\lbrace 0,14,26\rbrace,\lbrace 0,15,28\rbrace , \hfill \\ \lbrace 0,16,24\rbrace ,\lbrace 0,18,29\rbrace,\lbrace 0,19,23 \rbrace, \lbrace 0,20,25\rbrace,\hfill \\ \lbrace 0,21,27\rbrace,\lbrace 0,22,31\rbrace \hfill \end{gathered}$ \\
\hline
  $12$ & $73$ & $73,74,75,\ldots$ &  $803$ & $4$  & \textbf{PDF} & $\begin{gathered} \lbrace 0,1,36\rbrace,\lbrace 0,11,26\rbrace,\lbrace 0,6,29\rbrace,\lbrace 0,12,30\rbrace , \hfill \\ \lbrace 0,7,21\rbrace ,\lbrace 0,9,22\rbrace,\lbrace 0,8,33 \rbrace, \lbrace 0,10,34\rbrace,\hfill \\ \lbrace 0,3,31\rbrace,\lbrace 0,5,32\rbrace,\lbrace 0,2,19\rbrace,\lbrace 0,4,20\rbrace \hfill \end{gathered}$ \\
\hline
\end{tabular}
\endgroup
\label{Tab1}
\vspace{-.7em}
\end{table}
\begin{table}[ht]
\setlength{\tabcolsep}{.7 pt}
\centering
\caption{$(v,k=4,\lambda=1)$ Difference Families That Result In $1$-level Type-IV QC-LDPC Codes With Girth 6.
($N$ Is The Accurate Lifting Degree. $t$ Is The Number Of Blocks. $\mathbf{Dim}$ And $d_\mathbf{min}$ Are Respectively Dimension And Minimum Distance Of The Resulting Code Constructed By Smallest Accurate Lifting Degree $N$.)}
\vspace{-.7em}
\begingroup\fontsize{5 pt}{4 pt}
\begin{tabular}{|@{}c@{}|@{}c@{}|@{}l@{}|@{}c@{}|@{}c@{}|@{}c@{}|@{}l@{}|}
\hline
 $t$ & $v$ & ~~~~~~~~$N$ & $\mathbf{Dim}$ &$d_\mathbf{min}$& Type & ~~~~~~~~~~~~Difference Set's \\
\hline
\hline  
$1$ & $13$ & $13,14,15,\ldots$ &  $1$ & $13$  &\textbf{PDF}& $\lbrace 0,2,5,6\rbrace$ \\
\hline
  $2$ & $26$ &\begin{tabular}{@{}l@{}}
   $26,29,30,31,$\\ $33,34,\ldots$ \end{tabular}&  $ 27$ & $8$  & \textbf{None} & $\lbrace 0,1,3,9\rbrace,\lbrace 0,4,11,16\rbrace$ \\
\hline
  $3$ & $37$ & $37,41,42,\ldots$ &  $75$ & $6$  &\textbf{DF}& $\lbrace 0,12,19,20\rbrace,\lbrace 0,2,13,16\rbrace, \lbrace 0,5,9,15\rbrace$ \\
\hline
  $4$ & $49$ & $49,50,51,\ldots$ &  $148$ & $8$  &\textbf{PDF}& $\begin{gathered} \lbrace 0,1,7,23\rbrace,\lbrace 0,2,14,19\rbrace,\lbrace 0,3,13,21\rbrace, \\ \lbrace 0,4,15,24\rbrace\hfill \end{gathered}$\\
\hline
  $5$ & $61$ &$61,62,63,\ldots$&  $245$ & $7$  &\textbf{PDF}& $\begin{gathered} \lbrace 0,1,8,28\rbrace,\lbrace 0,2,14,24\rbrace,\lbrace 0,3,18,29\rbrace, \\ \lbrace 0,4,17,23\rbrace ,\lbrace 0,5,21,30\rbrace\hfill \end{gathered}$ \\
\hline
  $6$ & $73$ & $73,74,75,\ldots$ &  $366$ & $6$  & \textbf{PDF} & $\begin{gathered} \lbrace 0,1,34,36\rbrace,\lbrace 0,3,18,30\rbrace,\lbrace 0,4,20,28\rbrace,\hfill \\ \lbrace 0,5,22,31\rbrace ,\lbrace 0,6,19,29\rbrace ,\lbrace 0,7,21,32\rbrace\hfill \end{gathered}$ \\
\hline
  $7$ & $85$ & $85,86,87,\ldots$ &  $511$ & $8$  & \textbf{PDF} & $\begin{gathered} \lbrace 0,2,41,42\rbrace,\lbrace 0,5,30,33\rbrace,\lbrace 0,11,31,35\rbrace,\hfill \\ \lbrace 0,12,26,34\rbrace ,\lbrace 0,13,29,36\rbrace ,\lbrace 0,17,32,38\rbrace,\hfill \\ \lbrace 0,18,27,37 \rbrace
  \hfill \end{gathered}$ \\
\hline
  $8$ & $97$ & $97,98,99,\ldots$ &  $680$ & $8$  &  \textbf{PDF} & $\begin{gathered} \lbrace 0,2,47,48\rbrace,\lbrace 0,6,33,38\rbrace,\lbrace 0,8,37,44\rbrace, \hfill \\ \lbrace 0,11,35,39\rbrace ,\lbrace 0,15,31,41\rbrace ,\lbrace 0,17,30,42\rbrace,\hfill \\ \lbrace 0, 19,22,40\rbrace, \lbrace 0, 20,34,43\rbrace
  \hfill \end{gathered}$ \\
\hline
  $8$ & $97$ & $97,99,100,\ldots$ &  $680$ & $8$  &  \textbf{QPDF} & $\begin{gathered} \lbrace 0,2,3,49\rbrace,\lbrace 0,4,27,40\rbrace,\lbrace 0,5,26,43\rbrace, \hfill \\ \lbrace 0,6,25,39\rbrace ,\lbrace 0,7,31,42\rbrace ,\lbrace 0,8,30,45\rbrace,\hfill \\ \lbrace 0,9,29,41 \rbrace, \lbrace 0, 10,28,44\rbrace
  \hfill \end{gathered}$ \\
\hline
  $9$ & $109$ & $109,110,111,\ldots$ &  $873$ & $7$  & \textbf{PDF} &$\begin{gathered} \lbrace 0,2,53,54\rbrace,\lbrace 0,7,39,43\rbrace,\lbrace 0,9,42,50\rbrace, \hfill \\ \lbrace 0,11,38,48\rbrace ,\lbrace 0,15,35,49\rbrace ,\lbrace 0,16,40,46\rbrace,\hfill \\ \lbrace 0, 19,31,44\rbrace, \lbrace 0, 23,28,45\rbrace,  \lbrace 0, 26,29,47\rbrace
  \hfill \end{gathered}$ \\
\hline
  $10$ & $121$ & $121,122,123,\ldots$ &  $1090$ & $\leq 8$  & \textbf{PDF} & $\begin{gathered} \lbrace 0,2,59,60\rbrace,\lbrace 0,7,43,49\rbrace,\lbrace 0,10,47,56\rbrace, \hfill \\ \lbrace 0,12,45,53\rbrace ,\lbrace 0,15,40,54\rbrace ,\lbrace 0,16,35,48\rbrace,\hfill \\ \lbrace 0, 17,44,55\rbrace, \lbrace 0, 24,29,50\rbrace,  \lbrace 0, 28,31,51\rbrace, \hfill \\
  \lbrace 0,30,34,52\rbrace \hfill \end{gathered}$ \\
\hline
  $12$ & $145$ & $145,146,147,\ldots$ &  $1596$ & $\leq 8$  & \textbf{PDF} & $\begin{gathered} \lbrace 0,2,71,72\rbrace,\lbrace 0,9,51,58\rbrace,\lbrace 0,12,57,68\rbrace, \hfill \\ \lbrace 0,14,54,67\rbrace ,\lbrace 0,18,50,66\rbrace ,\lbrace 0,19,55,65\rbrace,\hfill \\ \lbrace 0,23,47,62\rbrace, \lbrace 0,25,52,60\rbrace,  \lbrace 0,30,34,63\rbrace,\hfill \\  \lbrace 0,31,37,59\rbrace, \lbrace 0,38,43,64\rbrace, \lbrace 0,41,44,61\rbrace \hfill \end{gathered}$ \\
\hline
\end{tabular}
\endgroup
\label{Tab2}
\vspace{-.7em}
\end{table}

Considering the results in \cite{RMathon1,JHuang1,JDinitz1,CColbourn1} we reproduced many PDF's and QPDF's. Tables \ref{Tab1} and \ref{Tab2} contain these $(v,k,1)$ PDF's and QPDF's\footnote{For the cases that a PDF doesn't exist, we consider a QPDF one.}, respectively for $k=3$ and $k=4$. Also included in these tables are  accurate lifting degrees, dimensions and minimum distances of $1$-level Type-$k$ CPM-QC-LDPC codes with girth $6$ that are constructed based on the Construction 1 using PDF's and/or QPDF's. Note that all of the reported dimensions and minimum distances are calculated for those constructed codes with smallest lifting degrees. Based on the definition of a PDF (QPDF), CPM based QC codes of the form in Construction 1 are of smallest possible lengths when their lifting degrees are taken to be the smallest accurate ones in the tables. Moreover, from \cite{BVasic1} we know that codes constructed based on DF's have an upper bound minimum distance equal to $2k$. From the results we see that codes constructed based on PDF's (QPDF's) also may have such optimum minimum distances. Based on the theorem \ref{Theo2}, there are two exceptions for the existence of a PDF (QPDF) when $k=4$ and $t=2,3$. Using a computer search we found the smallest possible  integers where they don't form a DF but lead to a $1$-level Type-$4$ CPM-QC-LDPC code with girth $6$ for the case $t=2$. When $t=3$, we also found a DF. The results for both of these cases are provided in Table \ref{Tab2}.
\begin{table}[ht]
\setlength{\tabcolsep}{.7 pt}
\centering
\caption{Parameters And Performances Of Simulated Codes}
\vspace{-.7em}
\begingroup\fontsize{5 pt}{4 pt}
\begin{tabular}{|@{}c@{}|@{}c@{}|@{}c@{}|@{}c@{}|@{}c@{}|@{}c@{}|@{}c@{}|@{}c@{}|@{}c@{}|@{}c@{}|@{}c@{}|@{}c@{}|@{}c@{}|}
\hline
~ & $d_v$ & $d_c$ & $C4$ & $C6$ & $C8$ & $C10$ & $\mathbf{Len}$ & $\mathbf{Dim}$ & $d_\mathbf{min}$ & $\mathbf{g}$ & \textit{SNR}$/$\textit{BER} & \textit{SNR}$/$\textit{FER} \\
\hline
\specialrule{.15em}{.06em}{.06em} 
 $\mathcal{C}^*_1$ &$3$&$6$& $0$ & $3276$ &$11193$ &$26754$ &$546$ &$273$ &$6$&$6$ & $2/1.73e-04$ &$2/1.40e-02$ \\
\hline
$\mathcal{C}^*_3$ &$3$&$6$& $0$ & $728$ & $5096$ & $29484$ &$546$ &$275$ &$10$&$6$ & $2/4.82e-04$ &$2/7.54e-03$ \\
\hline
 $\mathcal{C} [29]$ &$3$&$6$& $0$ & $0$ & $0$ & $10920$ & $546$ & $275$ & $14$ & $10$ & $2/4.75e-03$ &$2/6.09e-02$ \\
 \specialrule{.15em}{.06em}{.06em}
  $\mathcal{C}^*_1$ & $3$ & $9$ & $0$ & $912$ & $7524$ & $40356$ & $57$ & $38$ & $6$ & $6$ & $6.3/9.57e-07$ &$6.3/8.91e-06$ \\
\hline
  $\mathcal{C} [8]$ & $3$ & $9$ & $0$ & $912$ & $7581$ & $40014$ & $57$ & $38$ &$4$& $6$ & $6.3/3.15e-06$ &$6.3/4.11e-05$ \\
\specialrule{.15em}{.06em}{.06em}
$\mathcal{C}^*_1$ & $3$ & $9$ & $0$ & $42240$ & $306240$ & $1739760$ & $3960$ & $2640$ & $6$ & $6$ & $2/9.21e-06$ &$2/5.46e-03$ \\
\hline
$\mathcal{C}^*_4$ & $3$ & $9$ & $0$ & $7920$ & $59730$ & $586740$ & $3960$ & $2640$ &$\leq \hspace{-0.1cm}12$& $6$ & $2/1.27e-04$ &$2/3.15e-03$ \\
\hline
$\mathcal{C}^*_M$ & $3$ & $9$ & $0$ & $3960$ & $44550$ & $458700$ & $3960$ & $2640$ &$\leq \hspace{-0.1cm}26$& $6$ & $2/4.63e-04$ &$2/1.09e-02$ \\
\hline
  $\mathcal{C} [25]$ & $3$ & $9$ & $0$ & $0$ & $7590$ & $111870$ & $3960$ & $2640$ &$\leq \hspace{-0.1cm}470$& $8$ & $2/3.55e-03$ &$2/9.70e-02$ \\
\specialrule{.15em}{.06em}{.06em}
$\mathcal{C}^*_1$ & $3$ &$15$& $0$ & $4340$ & $74865$ & $1101306$ &$155$ & $124$ & $6$ & $6$ & $6/1.96e-07$ &$6/5.41e-06$ \\
\hline
$\mathcal{C} [8]$ & $3$ &$15$&$0$  & $4340$ & $78120$ & $1062432$ & $155$ & $124$ & $4$ & $6$ & $6/7.92e-06$ &$6/2.95e-04$ \\
\specialrule{.15em}{.06em}{.06em}
 $\mathcal{C}^*_1$ &$3$&$18$& $0$ & $7548$ & $164391$ & $3189252$ & $222$ & $185$ & $4$ & $6$ & $6/2.90e-07$ & $6/1.37e-05$\\
\hline
 $\mathcal{C} [8]$ & $3$ &$18$& $0$ & $7548$ &$164169$ & $3192804$ & $222$ & $185$ & $6$ &$6$ & $6/9.85e-08$ &$6/3.75e-06$\\
\specialrule{.15em}{.06em}{.06em}
 $\mathcal{C}^*_1$ &$3$&$18$& $0$ & $8341$ & $171342$ & $3516786$ & $342$ & $285$ & $6$ & $6$ & $6.4/1.10e-9$ &$6.4/6.98e-08$ \\
\hline
 $\mathcal{C} [5]$ &$3$&$18$& $0$ & $5472$ & $193743$ & $3755844$ & $342$ & $287$ & $6$ & $6$ & $6.4/4.30e-09$&$6.4/3.07e-07$ \\
\specialrule{.15em}{.06em}{.06em}
$\mathcal{C}^*_1$ &$3$&$21$&$0$ &$12040$ &$316050$  & $7649442$ & $301$ & $258$ & $6$ & $6$  &$6.3/1.44e-08$ & $6.3/6.63e-07$\\
\hline
$\mathcal{C} [8]$ &$3$&$21$& $0$ &$12040$ & $316050$ & $7650216$ & $301$ & $258$ & $6$ & $6$ & $6.3/1.22e-08$ & $6.3/5.26e-07$\\
\specialrule{.15em}{.06em}{.06em}
$\mathcal{C}^*_2$ & $4$ & $6$ & $0$ & $3591$ & $23009$ & $115710$ & $399$ & $135$ & $22$ & $6$ & $4/7.79e-09$ &$4/4.99e-08$ \\
\hline
$\mathcal{C}^*_3$ & $4$ & $6$ & $0$ & $3828$ & $14784$ & $45474$ & $396$ & $137$ & $14$ & $6$ & $4/7.79e-08$ &$4/1.30e-06$ \\
\hline
$\mathcal{C} [26]$ & $4$ & $6$ & $0$ & $798$ & $5852$ & $75278$ & $399$ & $135$ & $24$ & $6$ & $4/5.22e-06$ & $4/4.12e-05$\\
\specialrule{.15em}{.06em}{.06em}
$\mathcal{C}^*_2$ & $4$ & $10$ & $0$ & $12882$ & $163134$ & $1755258$ & $570$ & $344$ & $12$ & $6$ & $3.5/2.99e-08$ & $3.5/1.11e-07$\\
\hline
$\mathcal{C} [32]$ & $4$ & $10$ & $0$ & $0$ & $86526$ & $1418160$ & $570$ & $345$ & $50$ & $8$ & $3.5/4.16e-06$ &$3.5/6.37e-05$ \\
\specialrule{.15em}{.06em}{.06em}
$\mathcal{C}^*_2$ & $4$ & $10$ & $0$ & $43053$ & $545211$ & $5866257$ & $1905$ & $1145$ & $\leq\hspace{-0.1cm} 12$ &$6$ & $2.5/1.01e-07$& $2.5/5.79e-06$ \\
\hline
$\mathcal{C}^*_3$ & $4$ & $10$ & $0$ & $14440$ & $251370$ & $3768270$ & $1900$ & $1143$ & $\leq\hspace{-0.1cm} 88$ & $6$ & $2.5/1.92e-06$ & $2.5/3.91e-05$ \\
\hline
$\mathcal{C} [30]$ & $4$ & $10$ & $0$ & $2286$ & $77724$ & $1489710$ & $1905$ & $1145$ & $\leq\hspace{-0.1cm} 127$ & $6$ & $2.5/7.42e-05$ &  $2.5/1.28e-03$\\
\hline
$\mathcal{C} [26]$ & $4$ & $10$ & $0$ & $4191$ & $86106$ & $1536192$ & $1905$ & $1145$ & $\leq \hspace{-0.1cm}286$ & $6$ & $2.5/7.01e-05$ &$2.5/1.17e-03$ \\
\specialrule{.15em}{.06em}{.06em}
$\mathcal{C}^*_1$ & $4$ & $24$ & $0$ & $96960$ & $3674560$ & -- & $960$ & $801$ & $8$ & $6$ & $4.4/2.48e-07$ &$4.4/1.23e-05$ \\
\hline
$\mathcal{C} [35]$ & $3.4$ & $20$ & $40$ & $16520$ & $617840$ & -- & $960$ & $800$ & $7$ & $4$ & $4.4/9.80e-07$ &$4.4/5.93e-05$ \\
\specialrule{.15em}{.06em}{.06em}
$\mathcal{C}^*_1$ & $4$ & $32$ & $0$ & $144336$ & $8696535$ & -- & $776$ & $680$ & $8$ & $6$ & $5.8/4.01e-09$ &$5.8/1.94e-07$ \\
\hline
$\mathcal{C} [6]$ & $4$ &$32$ &  $702$ & $134345$ & $7705459$ & -- & $776$ &$680$ &  $4$ & $4$ & $5.8/2.93e-09$ & $5.8/3.45e-07$\\
\specialrule{.15em}{.06em}{.06em}
$\mathcal{C}^*_1$ & $4$ & $32$ & $0$ & $490560$ & $26399296$ & -- & $3584$ & $3137$ & $8$ & $6$ & $4/5.68e-07$ &$4/5.50e-05$ \\
\hline
$\mathcal{C} [34]$ & $4$ & $32$ & $0$ & $140180$ & $9577506$ & -- & $3584$ & $3141$ & $\leq \hspace{-0.1cm}140$ & $6$ & $4/1.13e-05$ &$4/7.45e-04$ \\
\specialrule{.15em}{.06em}{.06em}
\end{tabular}
\endgroup
\label{Tab3}
\vspace{-.7em}
\end{table}

Table \ref{Tab3} contains degree distributions (i.e., column and row weights), multiplicities of simple cycles (of lengths $4,6,8$ and $10$), lengths, dimensions, minimum distances and meaningful BER/FER performances of all of the simulated codes including our constructed codes and those one in the literature. By meaningful BER/FER performances, we mean that we've reported BER's and also FER's of simulated codes at Signal to Noise Ratio (SNR) points in which we observed a meaningful differences between the performances of some comparable codes, i.e., codes with the same lengths and degree distributions. These simulations preformed on a desktop computer with $3.5$ GHz CPU and $16$ GB RAM based on a SP decoding algorithm with maximum number of iterations equal to $150$, for binary phase shift keying (BPSK) modulation over AWGN channel. For each simulation point, fifty block errors are generated.

Codes denoted by $\mathcal{C}^*_1$ in Table \ref{Tab3} are $(d_v=k,d_c)$-regular $1$-level Type-$k$ CPM-QC-LDPC codes ($k=3,4$) with girth $6$ that are constructed based on Construction 1 using a $(v,k,1)$ PDF (or QPDF) reported in Tables \ref{Tab1} and \ref{Tab2}. Lengths of these codes depend on which one of lifting degrees are selected that is \begingroup\fontsize{9.5pt}{11pt}$\mathbf{Len}(\mathcal{C}^*_1)=N\frac{d_c}{d_v}=Nt$\endgroup. Codes denoted by $\mathcal{C}^*_2$ in Table \ref{Tab3} are $(d_v=4,d_c)$-regular $2$-level Type-$2$ CPM-QC-LDPC codes with girth $6$ that are constructed based on Construction 2 using a $(v,4,1)$ PDF (or QPDF) reported in Table \ref{Tab2}. Indeed, to reach a $(4,d_c)$-regular code $\mathcal{C}^*_2$, we applied $2$-CDT on the exponent matrix of a $(4,2d_c)$-regular code $\mathcal{C}^*_1$ in a way that we dispersed four elements in each column of exponent matrix $P^*_1$ by preserving two smallest ones to be in the first level of corresponding column in exponent matrix $P^*_2$. Clearly the next two elements must be considered as the elements in the second level of the specified column of $P^*_2$. See below example. 
\begin{Example}\label{Ex5}
Let $N=114$ and consider exponent matrix $P^*_1$ \begingroup\fontsize{9pt}{11pt}
$$\left[ \begin{array}{@{}c@{}c@{}c@{}c@{}c@{}c@{}}
(0,1,8,28)\;\;\; & (0,2,14,24)\;\;\; & (0,3,18,29)\;\;\; &(0,4,17,23) \;\;\; &(0,5,21,30)\end{array}\right]$$
\endgroup where, it's elements are given to be blocks of PDF in Table \ref{Tab2} with $t=5$. By applying Construction 1 on $P^*_1$, it results in a $(4,20)$-regular $1$-level Type-$4$ CPM-QC-LDPC code $\mathcal{C}^*_1$ with girth 6 and length $570$. Also, by applying $2$-CDT on $P^*_1$ we find exponent matrix $P^*_2$ as below: \begingroup\fontsize{9pt}{11pt}
$$\left[ \begin{array}{@{}c@{}c@{}c@{}c@{}c@{}c@{}}
(0,1)\;\;\; & (0,2)\;\;\; & (0,3)\;\;\; &(0,4) \;\;\; &(0,5)\\
(8,28)\;\;\; & (14,24)\;\;\; & (18,29)\;\;\; &(17,23) \;\;\; &(21,30)\end{array}\right]$$
\endgroup
where, it results in a $(4,10)$-regular $2$-level Type-$2$ CPM-QC-LDPC code $\mathcal{C}^*_2$ with girth $6$ and length $570$ reported in Table \ref{Tab3}.
\end{Example}
Codes denoted by $\mathcal{C}^*_3$ in Table \ref{Tab3} are $(d_v=k,d_c)$-regular $k$-level Type-$1$ CPM-QC-LDPC codes ($k=3,4$) with girth $6$ that are constructed based on Construction 2 using a $(v,k,1)$ PDF (or QPDF) reported in Tables \ref{Tab1} and \ref{Tab2}. Indeed, to reach a $(d_v=k,d_c)$-regular code $\mathcal{C}^*_3$, we applied $k$-CDT on exponent matrix of a $(k,kd_c)$-regular code $\mathcal{C}^*_1$. We dispersed $k$ elements in each column of exponent matrix $P^*_1$ starting from the smallest element, downward and in an increasing order, to put the $j$-th smallest element ($1\leq j\leq k$) in $j$-th level of corresponding column of exponent matrix $P^*_3$. See below example. 
\begin{Example}\label{Ex6}
Let $N=190$ and consider $1\times 10$ exponent matrix $P^*_1$, where, it's elements are given to be blocks of PDF in Table \ref{Tab2} with $t=10$. By applying Construction 1 on $P^*_1$, it results in a $(4,40)$-regular $1$-level Type-$4$ CPM-QC-LDPC code $\mathcal{C}^*_1$ with girth 6 and length $1900$. Also, by applying $4$-CDT on $P^*_1$ we find exponent matrix $P^*_3$ as below: \begingroup\fontsize{9pt}{11pt}
$$\left[ \begin{array}{@{}c@{}c@{}c@{}c@{}c@{}c@{}c@{}c@{}c@{}c@{}}
0\;\;\;\; & 0\;\;\;\; & 0\;\;\;\; &0\;\;\;\; &0\;\;\;\; &0\;\;\;\; &0\;\;\;\; &0\;\;\;\; &0\;\;\;\; &0\\
2\;\;\;\; & 7\;\;\;\; & 10\;\;\;\; &12\;\;\;\; &15\;\;\;\; &16\;\;\;\; &17\;\;\;\; &24\;\;\;\; &28\;\;\;\; &30\\
59\;\;\;\; & 43\;\;\;\; & 47\;\;\;\; &45\;\;\;\; &40\;\;\;\; &35\;\;\;\; &44\;\;\;\; &29\;\;\;\; &31\;\;\;\; &34\\
60\;\;\;\; & 49\;\;\;\; & 56\;\;\;\; &53\;\;\;\; &54\;\;\;\; &48\;\;\;\; &55\;\;\;\; &50\;\;\;\; &51\;\;\;\; &52
\end{array}\right],$$
\endgroup
where, it results in a $(4,10)$-regular $4$-level Type-$1$ CPM-QC-LDPC code $\mathcal{C}^*_3$ with girth $6$ and length $1900$ reported in Table \ref{Tab3}.
\end{Example}
Of structures of our constructed codes in Table \ref{Tab3}, $\mathcal{C}^*_4$ and $\mathcal{C}^*_M$ are still unknown. These codes are $(3,9)$-regular $4$-level Type-$1$ CPM-QC-LDPC codes with lifting degree $N=330$, girth $6$ and length $3960$. $\mathcal{C}^*_4$ is a code, where, it's exponent matrix $P^*_4$ is as follow:
\begingroup\fontsize{9pt}{11pt}
$$\left[ \begin{array}{@{}c@{}c@{}c@{}c@{}c@{}c@{}c@{}c@{}c@{}c@{}c@{}c@{}}
0\;\;\; & -\infty\;\;\; & 0\;\;\; &-\infty\;\;\; &0\;\;\; &-\infty\;\;\; &0\;\;\; &0\;\;\; &0\;\;\; &0\;\;\; &0\;\;\; &0\\
-\infty\;\;\; & 0\;\;\; &-\infty\;\;\; &0\;\;\; &-\infty\;\;\; &0\;\;\; &8\;\;\; &10\;\;\; &3\;\;\; &5\;\;\; &2\;\;\; &4\\
1\;\;\; & 11\;\;\; & 6\;\;\; &12\;\;\; &7\;\;\; &9\;\;\; &33\;\;\; &-\infty\;\;\; &31\;\;\; &-\infty\;\;\; &19\;\;\; &-\infty\\
36\;\;\; & 26\;\;\; & 29\;\;\; &30\;\;\; &21\;\;\; &22\;\;\; &-\infty\;\;\; &34\;\;\; &-\infty\;\;\; &32\;\;\; &-\infty\;\;\; &20
\end{array}\right].$$\endgroup
$P^*_4$ is designed by applying $4$-CDT on an exponent matrix $P^*_1$ that is constructed based on Construction 1 and blocks of PDF in Table \ref{Tab1} with $t=12$. $\mathcal{C}^*_M$ is a code, where, it's exponent matrix $P^*_M$ is as follow:
\begingroup\fontsize{9pt}{11pt}
$$\left[ \begin{array}{@{}c@{}c@{}c@{}c@{}c@{}c@{}c@{}c@{}c@{}c@{}c@{}c@{}}
0\;\;\; & -\infty\;\;\; & 0\;\;\; &-\infty\;\;\; &0\;\;\; &-\infty\;\;\; &0\;\;\; &0\;\;\; &0\;\;\; &0\;\;\; &0\;\;\; &0\\
-\infty\;\;\; & 9\;\;\; &-\infty\;\;\; &14\;\;\; &-\infty\;\;\; &19\;\;\; &23\;\;\; &25\;\;\; &30\;\;\; &31\;\;\; &38\;\;\; &41\\
71\;\;\; & 51\;\;\; & 57\;\;\; &54\;\;\; &50\;\;\; &55\;\;\; &47\;\;\; &-\infty\;\;\; &34\;\;\; &-\infty\;\;\; &43\;\;\; &-\infty\\
72\;\;\; & 58\;\;\; & 68\;\;\; &67\;\;\; &66\;\;\; &65\;\;\; &-\infty\;\;\; &60\;\;\; &-\infty\;\;\; &59\;\;\; &-\infty\;\;\; &61
\end{array}\right].$$\endgroup
$P^*_M$ is obtained in two steps: first, by applying $4$-CDT on an exponent matrix $P^*_1$ that is constructed based on Construction 1 and blocks of PDF in Table \ref{Tab2} with $t=12$; secondly, by {\em masking} the resultant exponent matrix in the first step. Note that masking technique is well-known for algebraic codes and we used the masking matrix in relation (4) of \cite{JLi1}.
\begin{table}[ht]
\setlength{\tabcolsep}{.7 pt}
\centering
\caption{Expected Number of Simple Cycles in a $\left(d_v,d_c \right)$-Regular Tanner Graph}
\vspace{-.7em}
\begingroup\fontsize{5 pt}{4 pt}
\begin{tabular}{|@{}c@{}|@{}c@{}|@{}c@{}|@{}c@{}|@{}c@{}|@{}c@{}|@{}c@{}|@{}c@{}|}
\hline
$d_v$ & $d_c$ & $\mathbf{Len}$($n^{\prime}$) & \textbf{Ro}($m^{\prime}$)& $\overline{C4}$ & $\overline{C6}$& $\overline{C8}$& $\overline{C10}$\\
\hline
$3$&$6$& $546$ & $273$ &$25$&$165$&$1230$&$9727$\\
\hline
$3$&$9$& $57$ & $19$ &$62$&$599$&$6201$&$64521$\\
\hline
$3$&$9$& $3960$ & $1320$ &$64$&$681$&$8162$&$104197$\\
\hline
$3$&$15$& $155$ & $31$ &$191$&$3355$&$64193$&$1265925$\\
\hline
$3$&$18$& $222$ & $37$ &$282$&$6084 $&$143416$&$3505416$\\
\hline
$3$&$18$& $342$ & $57$ &$285$&$6246 $&$151489$&$ 3850063$\\
\hline
$3$&$21$& $301$ & $43$ &$392$&$10000 $&$280318$&$8180916$\\
\hline
$4$&$6$& $399$ & $266$ &$56$&$557 $&$6202$&$73347$\\
\hline
$4$&$10$& $570$ & $228$ &$182$&$3242 $&$64809$&$1375905$\\
\hline
$4$&$10$& $1905$ & $762$ &$182$&$3269 $&$65943$&$1417074$\\
\hline
$4$&$24$& $960$ & $160$ &$1183$&$53771$&$2731187$&$147030214$\\
\hline
$4$&$32$& $776$ & $97$ &$2141$&$130067 $&$8794417$&$627521489$\\
\hline
$4$&$32$& $3584$ & $448$ &$2158$&$133191 $&$9228503$&$680517524$\\
\hline
$4$&$40$& $2400$ & $240$ &$3409$&$263690$&$22851973$&$2103513404$\\
\hline
\end{tabular}
\endgroup
\label{Tab4}
\vspace{-.7em}
\end{table}
\begin{figure}[ht]
\centering
\includegraphics[width=10cm, height=6cm]{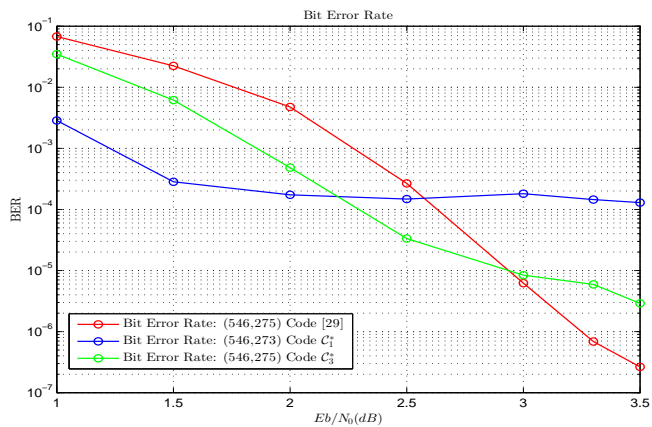}
\caption{Performance comparison of $(546,273)$ \textbf{QPDF} code $\mathcal{C}^*_1$ with $N=273$, and, $(546,275)$ code $\mathcal{C}^*_3$ with $N=91$ that are respectively, constructed based on Constructions 1 and 3, with $(546,275)$ codes reported in \cite{ATasdighi1} with $N=91$.
}\label{Fig:Reg36g6L546}
\end{figure}
\begin{figure}[ht]
\centering
\includegraphics[width=10cm, height=6cm]{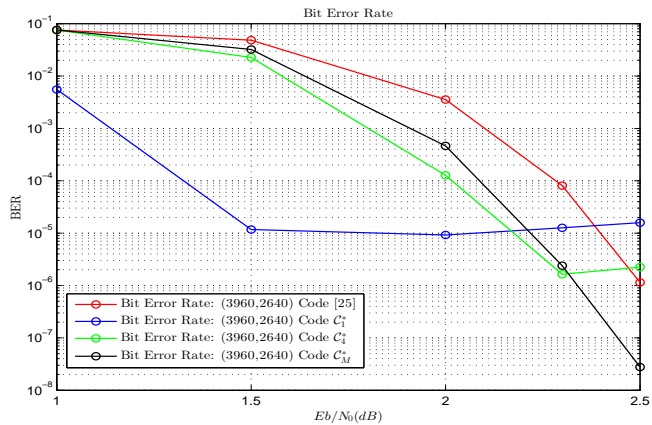}
\caption{Performance comparison of $(3960,2640)$ \textbf{QPDF} code $\mathcal{C}^*_1$ constructed based on Constructions 1 with $N=1320$, $(3960,2640)$ code $\mathcal{C}^*_4$ with $N=330$ and $(3960,2640)$ code $\mathcal{C}^*_M$ with $N=330$, with $(3960,2640)$ codes reported in \cite{JLi1} with $N=330$.
}\label{Fig:Reg39g6L3960}
\end{figure}

Table IV contains expected number of simple cycles of lengths $2k^{\prime}$ ($k^{\prime}=2,3,4,5$) in Tanner graph of a $(d_v,d_c)$-regular code with parity-check matrix of size $m^{\prime}\times n^{\prime}$. Values in this table are calculated from equation (\ref{Rela5}) and they are considered as reference values to compare the multiplicities of simple cycles of our simulated codes (in Table \ref{Tab3}) with them. 

In the following pros and cons of designing CPM-QC-LDPC codes using PDF's and/or QPDF's are notified. In addition, tailored to our justifications, performance curves for some of the simulated codes are sketched.

\textbf{1)} The proposed QC codes that are constructed based on Construction 1 have rates equal to $\frac{t-1}{t}$  ($2\leq t$) \cite{HPark2}. Although this rate cover many cases, it is not flexible when $t$ is small. Based on how one applies the column dispersion technique on a QC code of Construction 1, he/she is able to construct codes with more various rates, preserving their lengths, as well as, guaranteeing their girths to be at least $6$. For instance, $(4,10)$-regular codes $\mathcal{C}^*_2$'s and $\mathcal{C}^*_3$ in Table \ref{Tab3} have rates approximately equal to $R=0.6$, where,  $\frac{2-1}{2}< R<\frac{3-1}{3}$. 

\textbf{2)} For fixed parameters $t$, $d_v$ and $d_c$ the proposed QC codes that are constructed based on Construction 1 have various lengths considering accurate lifting degrees $N$ which are reported in Tables \ref{Tab1} and \ref{Tab2}. This fact is also true for the codes of Construction 2. In addition, based on remark \ref{Rem1} there may be some other (smaller) accurate lifting degrees for the codes of Construction 2. Moreover, given the smallest accurate lifting degrees for the codes of Construction 1, it is known \cite{HPark2} that they have smallest possible lengths among all CPM-QC-LDPC codes with girth 6. Array codes \cite{JFan1} and the ones presented in \cite{hagiwara1} are two different families of Type-$1$ CPM-QC-LDPC codes with girth $6$ which are famous for their small lengths. However, our counterpart codes of Construction 1 have much smaller lengths. Capability of our presented codes in achieving very high rates with small to moderate lengths have their own applications in designing error-correcting systems which require short packets.

\textbf{3)} Reported minimum distances for our codes of Construction 1 with minimum accurate lifting degrees (in Tables \ref{Tab1} and \ref{Tab2}) reveal that they have capabilities in reaching the upper bound minimum distances $2k$ of the constructed codes based on CDF's in \cite{BVasic1}. 

\textbf{4)} As it is expressed in \cite{HPark2}, we must pay heed to the use of codes of Construction 1 when we are taking their lifting degrees to be very larger than the minimum accurate one. Simulation results show that they may have no good error floor or even waterfall region in their performance curves. One of the main reasons for such unfavorable performances which it is not pointed in \cite{HPark2} is the existence of very high number of small inevitable cycles in Tanner graph of these codes. Due to the existence of small shift values in exponent matrix of such codes the probability of formation of many small inevitable cycles is high. As soon as we choose large lifting degrees, the multiplicities of such small cycles will linearly increase proportion to the selected lifting degrees. Multiplicities of all of the simple cycles of lengths up to $10$ of the simulated codes are reported in Table \ref{Tab3}. By comparing these multiplicities with their reference values in Table \ref{Tab4} for regular random LDPC codes, it is clear that codes of Construction 1 contain much more number of simple cycles in their Tanner graphs. Consider $(3,6)$ and $(3,9)$-regular codes with lifting degrees $273$ and $1320$ denoted by $\mathcal{C}^*_1$'s in Table \ref{Tab3}. Based on the results in Table \ref{Tab1}, minimum accurate lifting degree for a $(3,6)$-regular ($(3,9)$-regular)  \textbf{QPDF} code is $13$ ($19$). It means that, minimum accurate lifting degree for a $(3,6)$-regular ($(3,9)$-regular) \textbf{QPDF} code in Table \ref{Tab1} is about $21$ ($69$) times smaller than the one in Table \ref{Tab3}. By comparing the corresponding results in Tables \ref{Tab3} and \ref{Tab4} for $(3,6)$-regular ($(3,9)$-regular) code, one can see that their multiplicities of short simple cycles are much greater than the expected number of ones in counterpart random LDPC codes. For example, code $\mathcal{C}^*_1$ with lifting degree $273$ ($1320$) has $3276$ ($42240$) number of simple cycles of length $6$, while Tanner graph of a random LDPC code with similar parameters has nearly $165$ ($681$) number of simple cycles of length $6$, where, is about $20$ ($62$) times smaller than the ones in $\mathcal{C}^*_1$. Performances to $(3,6)$ and $(3,9)$-regular codes with lifting degrees $273$ and $1320$ are respectively, sketched in Figs. \ref{Fig:Reg36g6L546} and \ref{Fig:Reg39g6L3960} as blue curves. A very small sharp waterfall along with a high error floor in any of the curves are obvious.

\textbf{5)} Our method of constructing CPM-QC-LDPC codes using CDT, not only brings us a flexible scheme in designing new codes from primary codes with the same lengths, probably higher girths and higher minimum distances , but also, help us in preventing the collection of huge number of short simple cycles in Tanner graph of designed codes. Indeed, by dispersing each column, we are potentially prevent the union of small cycles. For instance, $(3,6)$-regular ($(3,9)$-regular) code $\mathcal{C}^*_3$ ($\mathcal{C}^*_4$) with lifting degree $91$ ($330$) is constructed using CDT. As it is shown in Table \ref{Tab3}, multiplicities of short simple cycles in Tanner graph of these codes are much smaller than ones in codes $\mathcal{C}^*_1$'s declared in previous part. Performances to $\mathcal{C}^*_3$ and $\mathcal{C}^*_4$ are respectively, sketched in Figs. \ref{Fig:Reg36g6L546} and \ref{Fig:Reg39g6L3960} as green curves. As we can see, waterfall regions of the curves of $\mathcal{C}^*_3$ and $\mathcal{C}^*_4$ are not as sharp as those one for $\mathcal{C}^*_1$'s, however, they have lower error floors thanks to their smaller number of short cycles. Also included in these figures, are the performances of the codes in \cite{ATasdighi1} and  \cite{JLi1} with similar rates and lengths to our designed codes. Since their multiplicities of short simple cycles are much smaller than our designed codes $\mathcal{C}^*_1$'s, $\mathcal{C}^*_3$ and $\mathcal{C}^*_4$ (see Table \ref{Tab3}), so, in contrast with our codes, they have no error floor at least down to BER$=10^{-6}$.  
With heuristic designs one can construct CPM-QC-LDPC codes of Construction 2 in which the multiplicities of small cycles are decreased gradually and with subtlety. For example, $(3,9)$-regular code $\mathcal{C}^*_M$ is of Construction 2 multiplied by a masking matrix in \cite{JLi1}, where, with this trick, we lowered the number of short simple cycles of the code even smaller than this number in $\mathcal{C}^*_4$ (see Table \ref{Tab3}). Performance of $\mathcal{C}^*_M$ is provided in Fig. \ref{Fig:Reg39g6L3960} as a black curve. This curve has sharper waterfall than the code in \cite{JLi1}, as well as, it has no error floor at least down to BER$=10^{-8}$.
\begin{figure}[ht]
\centering
\includegraphics[width=10cm, height=6cm]{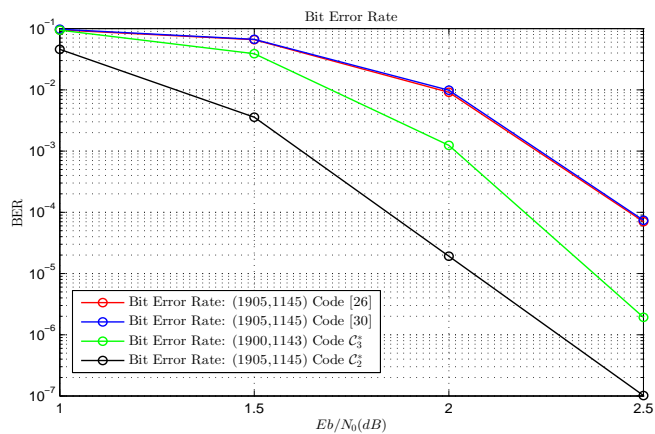}
\caption{Performance comparison of $(1905,1145)$ code $\mathcal{C}^*_2$ with $N=381$, and, $(1900,1143)$ code $\mathcal{C}^*_3$ with $N=190$ that are respectively, constructed based on Constructions 2 and 3, with $(1905,1145)$ codes reported in \cite{LZhang1,GZhang1} with $N=381$.
}\label{Fig:Reg410g6L1905}
\end{figure}
\begin{figure}[ht]
\centering
\includegraphics[width=10cm, height=6cm]{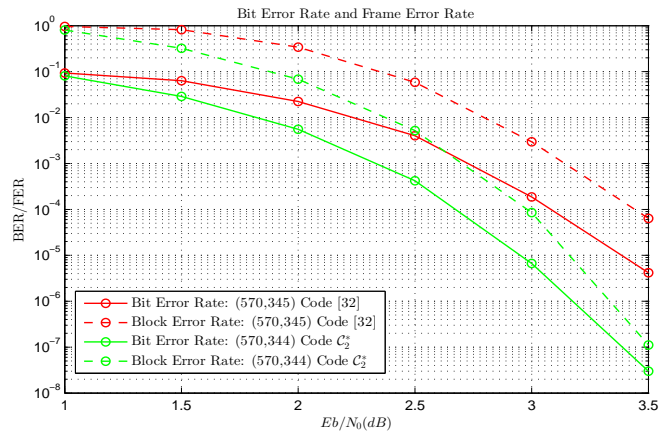}
\caption{Performance comparison of $(570,344)$ code $\mathcal{C}^*_2$ with $N=114$ based on Construction 2, with a $(570,345)$ code \cite{ATasdighi2} with $N=57$.
}\label{Fig:Reg410g6L570}
\end{figure}

\textbf{6)} We did simulate many $m$-level Type-$L$ CPM-QC-LDPC codes including our codes of Constructions 1 and 2, as well as, well-known codes in the literature, for lots of $m$'s, $L$'s and with various lengths and rates. Based on our reported experimental results in Table \ref{Tab3}, as well as, many other simulation results that we did not report them for the sake of brevity, slopes of performance curves of codes with the same rates, lengths and degree distributions are mainly influenced by two critical factors: minimum distances and multiplicities of short simple cycles of the codes. The positive role of minimum distance in performances of linear codes is overt. On the other hand, the larger the number of short simple cycles in Tanner graph of a code is, the sharper the waterfall of that code is. Also, the role of smaller cycles in promoting waterfall is much more important than the role of larger cycles. As an example, we compared multiplicities of short simple cycles in Tanner graph of four $(4,10)$-regular codes $\mathcal{C}^*_2$,  $\mathcal{C}^*_3$, a code in \cite{GZhang1} and a code in \cite{LZhang1}, respectively, of lengths $1905$, $1900$, $1905$ and $1905$ which are reported in Table \ref{Tab3}. From this table, we can see that number of short simple cycles in $\mathcal{C}^*_2$ is bigger than those one in $\mathcal{C}^*_3$. Furthermore, for each of the codes $\mathcal{C}^*_2$ and $\mathcal{C}^*_3$ this number is much bigger than the expected number of simple cycles in a $(4,10)$-regular random code of length $1905$ reported in Table \ref{Tab4}. On the other hand, multiplicities of short cycles in each one of the codes in \cite{GZhang1} and \cite{LZhang1} are more or less the same and are nearly equal to the expected number in a random code. The minimum founded upper bounds to the minimum distances of these codes are also provided in Table \ref{Tab3}. Curves to the BER performances of these codes are sketched in Fig. \ref{Fig:Reg410g6L1905}, where, the sharper waterfall is for our constructed code $\mathcal{C}^*_2$ with the most large number of short cycles in its Tanner graph, howbeit, it seems that $\mathcal{C}^*_2$ has smallest minimum distance in comparison with the others. The second sharp waterfall is for our code $\mathcal{C}^*_3$, where, Tanner graph of this code stands in second place to have large number of short cycles.
\begin{figure}[ht]
\centering
\includegraphics[width=10cm, height=6cm]{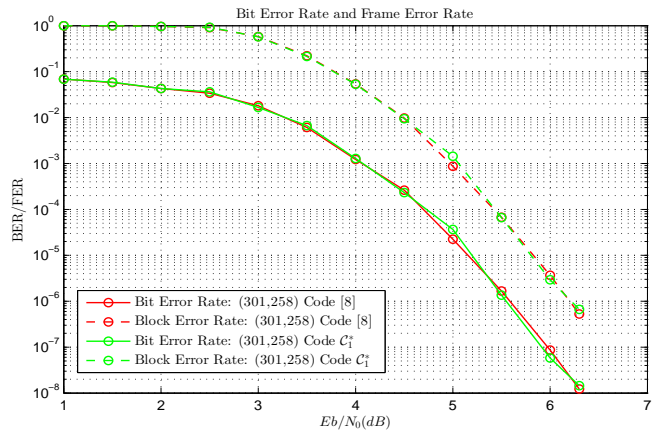}
\caption{Performance comparison of $(301,258)$ \textbf{QPDF} code $\mathcal{C}^*_1$ with $N=43$ based on Construction 1, with a $(301,258)$ \textbf{CDF} code \cite{BVasic1} with $N=43$.
}\label{Fig:Reg321g6L301}
\end{figure}
To explicit the importance of high number of short simple cycles in improving waterfall region of performance of a code, we consider two other $(4,10)$-regular codes $\mathcal{C}^*_2$ and symmetric code in \cite{ATasdighi2} both of length $570$ in Table \ref{Tab3}. From this table, it is clear that the number of short simple cycles of our constructed code $\mathcal{C}^*_2$ is much higher than this number for the symmetric code. Fig. \ref{Fig:Reg410g6L570} contain BER/FER curves to the performances of these codes, where, $\mathcal{C}^*_2$ outperforms the symmetric code in waterfall region at least down to BER$=10^{-8}$.  Note that symmetric code has a very large minimum distance compared to $\mathcal{C}^*_2$, however, it seems that its lack of short cycles of length $6$, as well as, its small number of other short cycles has more effects on declining the waterfall than that its high minimum distance affects on improving the waterfall.  
\begin{figure}[ht]
\centering
\includegraphics[width=10cm, height=6cm]{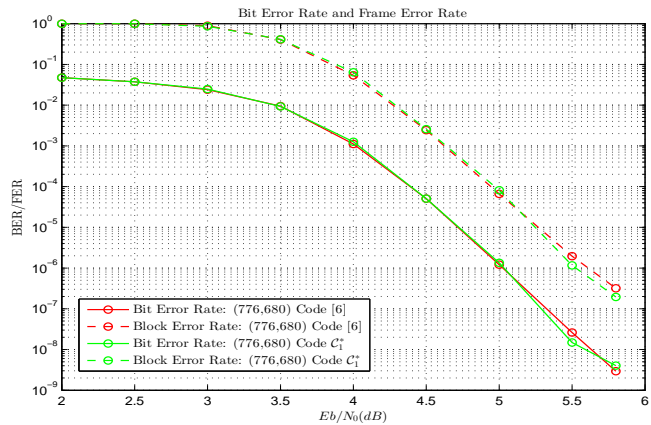}
\caption{Performance comparison of $(776,680)$ \textbf{QPDF} code $\mathcal{C}^*_1$ with $N=97$, with a $(776,680)$ PEG code \cite{XHu1}.
}\label{Fig:Reg432g6L776}
\end{figure}
Finally, to concurrently check the influences of minimum distances and number of short simple cycles on waterfall region of performances of short to moderate length CPM-QC-LDPC codes, we refer the readers to our  results of simulated codes in Tables \ref{Tab3} and \ref{Tab4}. Specially, those compared codes in Figs. \ref{Fig:Reg321g6L301}, \ref{Fig:Reg432g6L776} and \ref{Fig:Reg424g6L960}. For instance, an interesting result will come out if we compare two $(3,21)$-regular code $\mathcal{C}^*_1$ and the code in \cite{BVasic1} both with girths $6$ and lengths $301$. As we can see from the Table \ref{Tab3}, these codes have the same minimum distances, as well as, the same number of simple cycles of lengths $6$ and $8$. They have a very slice (venial) difference in their number of simple cycles of length $10$. Curves to BER/FER performances of these codes are drawn in Fig. \ref{Fig:Reg321g6L301}, where, we can see a perfect match between these curves down to BER$=10^{-8}$. This conformity among their performances is mainly due to their equal parameters, specially, their equal minimum distances, as well as, their close number of short simple cycles. 
\begin{figure}[ht]
\centering
\includegraphics[width=10cm, height=6cm]{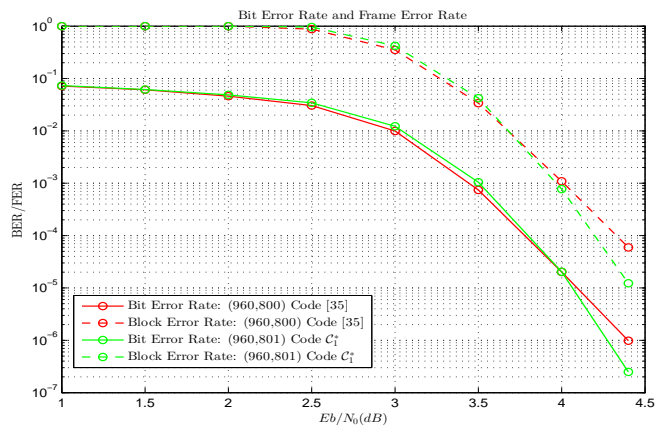}
\caption{Performance comparison of $(960,801)$ \textbf{PDF} code $\mathcal{C}^*_1$ with $N=160$, with a $(960,800)$ WiMAX code \cite{WiMax1}.
}\label{Fig:Reg424g6L960}
\end{figure}

\section{Conclusion} \label{Sec5}
In this paper, using CDF's a class of combinatoric-based QC-LDPC codes with $dv = 3$ or $4$ is reconsidered. In fact, taking the advantages of PDF's and QPDF's as two specific sub-classes of CDF's, Type-$L$ ($L=3,4$) parity-check matrices having girth 6, shortest possible length and consisting of a single row of circulants are constructed. In the sequel and picking a PDF (QPDF) out, a new flexible design technique named CDT is proposed. By new constructed codes based on CDT, not only we are able to preserve the girths also the minimum distances of new codes to be greater than or equal to the primary ones of underlying codes, but also we subtly can derive new codes with various Types and rates. 

One other privilege of constructing QC-LDPC codes using PDF's (QPDF's) was their high number of short simple cycles which in turn affects the performances of the codes. Considering the experimental results that reported in Tables \ref{Tab3} and \ref{Tab4}, we clarified that constructed codes based on PDF's (QPDF's), as well as, those new introduced ones using CDT mainly possess higher multiplicities of small cycles with respects to general LDPC codes. Comparing our proposed codes with numerous well-known existing regular LDPC codes of similar rates, degree distributions and lengths, we statistically demonstrated the impact of the existence of high number of simple cycles in Tanner graph of LDPC codes. This impact was positive for waterfall region but negative for error floor region of the code. Moreover, we illustrated the capability of the introduced CDT in designing codes with arbitrary number of short cycles, where they were prevented to have precocious error floors. Finally, we presented several curves to our simulation results which they verified our assertions.

\section*{Acknowledgment}

The authors would like to thank...



\begin{thebibliography}{1}
\bibitem{RMathon1}
R.~Mathon, ``Constructions for cyclic steiner 2 designs,'' {\em Annals of Discrete Mathematics},  vol.~34, pp.~353--362, 1987.

\bibitem{JHuang1}
J.~H.~Huang and S.~S.~Skiena, ``Gracefully labeling prisms,'' {\em Ars. Combin.}, vol.~38, pp.~225--242, 1994.

\bibitem{DMacKay1}
D.~J.~C.~MacKay and R.~M.~Neal, ``Near Shannon limit performance of low density parity-check codes,'' {\em Electron. Lett.}, vol.~32, no.~18, pp.~1645--1646, Aug.~1996.

\bibitem{JDinitz1}
J.~H.~Dinitz and P.~Rodney, ``Disjoint difference families with block size 3," {\em Utilitas Math}, vol.~52, pp.~153--160, 1997.

\bibitem{JFan1}
J.~L.~Fan,``Array codes as low-density parity-check codes,'' {\em in
Proc. 2nd Int. Symp. Turbo Codes}, Brest, France, pp.~543--546, Sep. 2000.

\bibitem{XHu1}
X.~Y.~Hu, E.~Eleftheriou and D.~M.~Arnold, ``Progressive edge-growth Tanner graphs,'' {\em Global Telecommunications Conference}, Nov 2001.

\bibitem{Thrope1} J. Thorpe, ``Low-density parity-check (LDPC) codes constructed from protograph" {\em IPN Progr. Rep.}, 42-154, Aug. 2003, JPL

\bibitem{BVasic1}
B.~Vasic, ``Combinatorial constructions of low-density parity-check codes for iterative decoding,'' {\em IEEE Trans. Inf. Theory}, vol.~50, no. 6, pp.~1156--1176, June 2004.

\bibitem{MFossorier1}
M.~P.~Fossorier, ``Quasi-cyclic low-density parity-check codes from circulant permutation matrices,'' {\em IEEE Trans. Inf. Theory}, vol.~50, no. 8, pp.~1788--1793, Aug. 2004.

\bibitem{Smarandache2}
R.~Smarandache and P.~O.~Vontobel, ``On regular quasi-cyclic LDPC codes from binomials," {\em in Proc. IEEE Int. Symp. Inf. Theory}, p.~274, 2004.

\bibitem{Divsalar1} D. Divsalar, D. Sam, and J. Christopher, ``Construction of protograph LDPC codes with linear minimum distance," {\em in Proc. IEEE Int. Symp. Inf. Theory}, pp.~664--668. IEEE, 2006.

\bibitem{sullivan1}
M.~E.~O'Sullivan, ``Algebraic construction of sparse matrices with large girth,'' {\em IEEE Trans. Inf. Theory}, vol.~52, no.~2, pp.~718--727, Feb. 2006.

\bibitem{CColbourn1}  
C.~J.~Colbourn and J.~H.~Dinitz, ``The CRC handbook of combinatorial designs,'' {\em Chapman \& Hall/CRC}, Second Edition, 2006.

\bibitem{Lally1} 
K.~Lally, ``Explicit construction of type-II QC LDPC codes with girth at least 6," {\em IEEE ISIT 2007}, Nice, France, pp.~2371--2375, June 2007.

\bibitem{hagiwara1}
M.~Hagiwara, M. P. C.~Fossorier, T.~Kitagawa, and H.~Imai, ``Smallest size of circulant matrix for regular (3,L) and (4,L) quasi-cyclic LDPC codes with girth 6,'' {\em IEICE Trans. Fundamentals of Electronics, Communications and Computer Sciences}, vol.~E92.A, no.~11, pp.~2891--2894, Nov. 2009.

\bibitem{ZChen1}
Z.~Chen, D.~Wu and P.~Fan, ``Applications of additive sequence of permutations,'' {\em Discrete Mathematics}, vol.~309, pp.~6459--6463, 2009.

\bibitem{MBaldi1}
M.~Baldi, M.~Bianchi, G.~Cancellieri, F.~Chiaraluce and T.~Klove, ``On the generator matrix of array LDPC codes,'' {\em in Proc. SoftCOM}, Split, Croatia, pp.~1--5, Sep. 2012.

\bibitem{bocharova1}
I. E.~Bocharova, F.~Hug, R.~Johannesson, B. D.~Kudryashov, and R. V.~Satyukov, ``Searching for voltage graph-based LDPC tailbiting codes with large girth,'' {\em IEEE Trans. Inf. Theory}, vol.~58, no. 4, pp.~2265--2279, Apr. 2012.

\bibitem{Smarandache1}
R.~Smarandache and P.~O.~Vontobel, ``Quasi-cyclic LDPC codes: influence of proto- and Tanner-graph structure on minimum hamming distance upper bounds," {\em  IEEE Trans. Inf. Theory}, vol.~58, no.~2, pp.~585--607, 2012.

\bibitem{HPark1}
H.~Park, S.~Hong, J.~S.~No and D.~J.~Shin, ``Design of multiple-edge protographs for QC LDPC codes avoiding short inevitable cycles," {\em  IEEE Trans. Inf. Theory}, vol.~59, no.~7, pp.~4598--4614, 2012.

\bibitem{karimi1}
M.~Karimi and A.~Banihashemi, ``On the girth of quasi cyclic protograph LDPC codes,'' {\em IEEE Trans. Inf. Theory}, vol.~59, no.~7, pp.~4542--4552, July 2013.

\bibitem{HPark2}  H.~Park, S.~Hong, J.~S.~No and D.~J.~Shin, ``Construction of high-rate regular quasi-cyclic LDPC codes based on cyclic difference families,'' {\em IEEE Transaction on Communications}, vol.~61, no.~8, pp.~3108--3113, Aug. 2013.

\bibitem{XJiao1}
X.~Jiao, and M.~Jianjun, ``Probabilistic analysis of cycles in random Tanner graphs," {\em Signal Processing, Communication and Computing (ICSPCC), 2013 IEEE International Conference on.}, pp.~1-5, Aug.~2013.

\bibitem{HLiu1}
H.~Liu, S.~Yang, G.~Deng and J.~Chen, ``More on the minimum distance of array LDPC codes," {\em IEEE Communications Letters}, vol.~18, no.~9, pp.1479--1482, Sept. 2014.

\bibitem{JLi1}  J.~Li, K.~Liu, S.~Lin and K.~Abdel-Ghaffar, ``Algebraic quasi-cyclic LDPC codes: construction, low error-floor, large girth and a reduced-complexity decoding scheme,'' {\em IEEE Transaction on Communications}, vol.~62, no.~8, pp.~2626--2637, Aug. 2014.

\bibitem{LZhang1} 
L.~Zhang, B.~Li, and L.~Cheng, ``Construction of type-II QC LDPC codes based on perfect cyclic difference set," {\em Chin. J. Electron}, vol.~24, no.~1, pp.~146--151, 2015.

\bibitem{SRanganathan1}
S. V. S.~Ranganathan, D.~Divsalar, and R. D.~Wesel, ``On the girth of (3,L) quasi-cyclic LDPC codes based on complete protographs,'' in {\em Proc. IEEE Int. Symp. Information Theory}, June 2015.

\bibitem{CSchoeny1} 
C.~Schoeny, ``Quasi-Cyclic Non-Binary LDPC Codes for MLC NAND Flash Memory,'' {\em http://nvmw.ucsd.edu/2015/speakerinfo/68}, Aug. 2015.

\bibitem{ATasdighi1}
A.~Tasdighi, A. H.~Banihashemi, and M. R.~Sadeghi, ``Efficient search of girth-optimal QC-LDPC codes,'' {\em IEEE Trans. Inf. Theory}, vol.~62, no.~4, pp. 1552--1564, April 2016.

\bibitem{GZhang1} 
G.~Zhang, ``Type-II quasi-cyclic low-density parity-check codes from Sidon sequences,'' {\em Electronics Letter}, vol.~52, no.~5, pp.~367--369, March 2016.

\bibitem{MBaldi2}
M.~Baldi, M.~Battaglioni, F.~Chiaraluce, and G.~Cancellieri, ``Time-invariant spatially coupled low-density parity-check codes with small constraint length,'' in \emph{Proc. IEEE BlackSeaCom 2016}, Varna, Bulgaria, Jun. 2016.

\bibitem{ATasdighi2}  A. Tasdighi, A. H. Banihashemi and M. R. Sadeghi, ``Symmetrical constructions for regular girth-8 QC-LDPC codes,'' {\em IEEE Transaction on Communications}, vol.~14, no.~8, pp. ?--?, Oct. 2016.
  
\bibitem{MAGMA} Online. http://magma.maths.usyd.edu.au/magma/
\bibitem{McKayWeb} Online. http://www.inference.phy.cam.ac.uk/mackay/codes/data.html\#s22
\bibitem{WiMax1} Online. http://www.uni-kl.de/en/channel-codes/ml-simulation-results/
\end{thebibliography}
\end{document}